\documentclass[11pt, a4paper]{article}
\def\PrintMode{0}%

\usepackage{geometry}
\usepackage{times}

\usepackage[T1]{fontenc}
\usepackage{graphicx}
\usepackage{authblk}
\usepackage{algorithm,algpseudocode}
\usepackage{xspace}
\usepackage[utf8]{inputenc}
\usepackage[noadjust]{cite}
\usepackage{amsmath, amssymb, amsfonts, amsthm}
\usepackage{enumerate}
\usepackage{xcolor}
\usepackage[linktocpage=true,pagebackref=true,colorlinks,linkcolor=magenta,citecolor=blue,bookmarks,bookmarksopen,bookmarksnumbered]{hyperref}
\usepackage{thmtools}
\usepackage{paralist}
\usepackage{gensymb}
\usepackage{rotating}
\usepackage{mdframed}
\usepackage{multirow}
\usepackage{verbatim} 
\usepackage{mleftright}
\usepackage{tablefootnote}

\ifnum\PrintMode=1
\usepackage{savetrees}

\else
\fi

\usepackage[capitalize]{cleveref}%

\newtheorem{theorem}{Theorem}[section]
\newtheorem{lemma}[theorem]{Lemma}

\newtheorem{claim}[theorem]{Claim}

\theoremstyle{definition}
\newtheorem{definition}[theorem]{Definition}

\theoremstyle{remark}
\newtheorem{remark}[theorem]{Remark}

\renewcommand*\backref[1]{\ifx#1\relax \else (cit. on p. #1) \fi} %

\algtext*{EndWhile}%
\algtext*{EndIf}%
\algtext*{EndFor}%
\algtext*{EndFunction}%

\algnewcommand{\IfThen}[2]%
{\State \algorithmicif\ #1\ \algorithmicthen\ #2}

\makeatletter
\def\moverlay{\mathpalette\mov@rlay}
\def\mov@rlay#1#2{\leavevmode\vtop{%
		\baselineskip\z@skip \lineskiplimit-\maxdimen
		\ialign{\hfil$\m@th#1##$\hfil\cr#2\crcr}}}
\newcommand{\charfusion}[3][\mathord]{
	#1{\ifx#1\mathop\vphantom{#2}\fi
		\mathpalette\mov@rlay{#2\cr#3}
	}
	\ifx#1\mathop\expandafter\displaylimits\fi}
\makeatother

\newcommand{\To}{\rightsquigarrow}

\newcommand{\bfA}{\mathbf{A}}
\newcommand{\bfB}{\mathbf{B}}
\newcommand{\bfC}{\mathbf{C}}

\newcommand{\bfF}{\mathbf{F}}
\newcommand{\bfI}{\mathbf{I}}
\newcommand{\bfM}{\mathbf{M}}
\newcommand{\bfN}{\mathbf{N}}
\newcommand{\bfZ}{\mathbf{Z}}

\newcommand{\bfb}{\mathbf{b}}

\newcommand{\bfe}{\mathbf{e}}
\newcommand{\bfp}{\mathbf{p}}
\newcommand{\bfu}{\mathbf{u}}
\newcommand{\bfv}{\mathbf{v}}

\newcommand{\caD}{\mathcal{D}}
\newcommand{\caH}{\mathcal{H}}
\newcommand{\caR}{\mathcal{R}}

\newcommand{\dbF}{\mathbb{F}}

\newcommand{\dbZ}{\mathbb{Z}}
\newcommand{\sfD}{\mathsf{D}}
\newcommand{\sfL}{\mathsf{L}}
\newcommand{\sfR}{\mathsf{R}}
\newcommand{\sfT}{\mathsf{T}}
\newcommand{\sfU}{\mathsf{U}}

\newcommand{\Tin}{T^{\sf in}}
\newcommand{\Tout}{T^{\sf out}}

\newcommand{\MM}{\mathsf{MM}}
\newcommand{\SA}{\mathsf{SA}}

\newcommand{\Fast}{\mathsf{Fast}}
\newcommand{\Extend}{\mathsf{Extend}}
\newcommand{\Dstart}{\caD^{\mathsf{start}}}
\newcommand{\Dfinal}{\caD^{\mathsf{final}}}

\DeclareMathOperator{\adj}{adj}

\DeclareMathOperator{\diag}{diag}
\DeclareMathOperator{\cdeg}{cdeg}

\def\ShowAuthNotes{1}
\ifnum\ShowAuthNotes=1
\newcommand{\authnote}[2]{\ \\ \textcolor{red}{\parbox{0.9\linewidth}{[{\footnotesize {\bf #1:} { {#2}}}]}}\newline}
\else
\newcommand{\authnote}[2]{}
\fi

\begin{document}
	
\title{Constructing a Distance Sensitivity Oracle in $O(n^{2.5794}M)$ Time}

\author[1]{Yong Gu \thanks{\href{mailto:guyong12@mails.tsinghua.edu.cn}{\texttt{guyong12@mails.tsinghua.edu.cn}}.}}
\author[1]{Hanlin Ren \thanks{\href{mailto:rhl16@mails.tsinghua.edu.cn}{\texttt{rhl16@mails.tsinghua.edu.cn}}.}}

\affil[1]{Institute for Interdisciplinary Information Sciences, Tsinghua University}

\maketitle

\begin{abstract}
	We continue the study of \emph{distance sensitivity oracles} (DSOs). Given a directed graph $G$ with $n$ vertices and edge weights in $\{1, 2, \dots, M\}$, we want to build a data structure such that given any source vertex $u$, any target vertex $v$, and any failure $f$ (which is either a vertex or an edge), it outputs the length of the shortest path from $u$ to $v$ not going through $f$. Our main result is a DSO with preprocessing time $O(n^{2.5794}M)$ and constant query time. Previously, the best preprocessing time of DSOs for directed graphs is $O(n^{2.7233}M)$, and even in the easier case of undirected graphs, the best preprocessing time is $O(n^{2.6865}M)$ [Ren, ESA 2020]. One drawback of our DSOs, though, is that it only supports distance queries but not path queries.
	
	Our main technical ingredient is an algorithm that computes the inverse of a degree-$d$ polynomial matrix (i.e.~a matrix whose entries are degree-$d$ univariate polynomials) modulo $x^r$. The algorithm is adapted from [Zhou, Labahn, and Storjohann, {\it Journal of Complexity}, 2015], and we replace some of its intermediate steps with faster rectangular matrix multiplication algorithms.
	
	We also show how to compute \emph{unique} shortest paths in a directed graph with edge weights in $\{1, 2, \dots, M\}$, in $O(n^{2.5286}M)$ time. This algorithm is crucial in the preprocessing algorithm of our DSOs. Our solution improves the $O(n^{2.6865}M)$ time bound in [Ren, ESA 2020], and matches the current best time bound for computing all-pairs shortest paths.
\end{abstract}

\section{Introduction}
In this paper, we consider the problem of constructing a \emph{distance sensitivity oracle} (DSO). A DSO is a data structure that preprocesses a directed graph $G = (V, E)$ with $n$ vertices and $m$ edges, and supports queries of the following form: Given a source vertex $u$, a target vertex $v$, and a failure $f$ (which can be either a vertex or an edge), output the length of the shortest path from $u$ to $v$ that does not go through $f$.

One motivation for constructing DSOs is the fact that real-life networks often suffer from failures. Consider a communication network among $n$ servers. When a server $u$ wants to send a message to another server $v$, the most efficient way would be to send the message along the shortest path from $u$ to $v$. However, if a failure happens in a server or a link between two servers, we would need to recompute the shortest path with the failure taken into account. It may be too slow to compute the shortest path from scratch each time a failure happens. A better solution is to construct a DSO for the communication network, and invoke the query algorithm of the DSO whenever a failure happens.

\subsection{Related Work}
The problem of constructing DSOs has received a lot of attention in the literature. A na\"ive solution is to precompute the answers for every possible query $(u, v, f)$, but it requires $\Omega(n^2m)$ space to store this DSO. Demetrescu et al.~\cite{DemetrescuTCR08} constructed a DSO with $O(n^2\log n)$ space that answers a query in constant time. However, the preprocessing time of the DSO in \cite{DemetrescuTCR08} is $O(mn^2 + n^3\log n)$, which is inefficient for large networks. Subsequently, Bernstein and Karger improved the preprocessing time to $\tilde{O}(n^2\sqrt{m})$ \cite{BernsteinK08}, and finally $\tilde{O}(mn)$ \cite{BernsteinK09}.\footnote{$\tilde{O}$ hides $\operatorname{polylog}(n)$ factors.} The preprocessing time $\tilde{O}(mn)$ matches the current best time bound for the easier problem of computing \emph{all-pairs shortest paths} (APSP), and it is conjectured that APSP requires $mn^{1-o(1)}$ time \cite{LincolnWW18}. In this sense, the $\tilde{O}(mn)$ time bound of \cite{BernsteinK09} is optimal. Duan and Zhang \cite{DuanZ17a} improved the space complexity of the DSO to $O(n^2)$, eliminating the last $\log n$ factor, while preserving constant query time and $\tilde{O}(mn)$ preprocessing time.

However, for dense graphs (i.e.~$m=\Theta(n^2)$) with edge weights in $[-M, M]$, it is possible to compute APSP in time faster than $\tilde{O}(mn) = \tilde{O}(n^3)$. The best APSP algorithm for undirected graphs runs in $\tilde{O}(n^\omega M)$ time \cite{Seidel95, ShoshanZ99}, and the best APSP algorithm for directed graphs runs in $O(n^{2.5286}M)$ time \cite{AlonGM97, Zwick02}. (Here $\omega < 2.3728596$ is the exponent of matrix multiplication \cite{CW90, Sto10, Wil12, LeGall, AlmanW21}.) Therefore, it is natural to ask whether one can beat $\tilde{O}(n^3)$ preprocessing time for DSOs in this regime.

The answer turned out to be \emph{yes}. Weimann and Yuster \cite{WeimannY13} showed that for any constant $0 < \alpha < 1$, there is a DSO with $\tilde{O}(n^{1-\alpha + \omega}M)$ preprocessing time and $\tilde{O}(n^{1+\alpha})$ query time. Subsequently, Grandoni and Williams \cite{GrandoniW12} showed that for any constant $0 < \alpha < 1$, there is a DSO with $\tilde{O}(n^{\omega + 1/2}M + n^{\omega + \alpha(4-\omega)}M)$ preprocessing time and $\tilde{O}(n^{1-\alpha})$ query time. Recently, Chechik and Cohen \cite{ChechikC20} constructed the first DSO that achieves both sub-cubic ($O(n^{2.873}M)$) preprocessing time and poly-logarithmic query time simultaneously. For the case that edge weights are positive, Ren \cite{Ren20} improved the previous results by presenting a much simpler DSO with $\tilde{O}(n^{2.7233}M)$ preprocessing time and constant query time.

Note that most DSOs mentioned above are randomized. Recently, there are also some efforts on derandomizing these DSOs, see e.g.~\cite{AlonCC19, KarthikP21}.

\subsection{Our Results}

Our main result is an improved DSO for directed graphs with integer edge weights in $[1, M]$. In particular, our DSO has preprocessing time $O(n^{2.5794}M)$ and constant query time.

\begin{theorem}[Main]\label{thm:main}
	Given as input a directed graph $G=(V, E)$ with edge weights in $\{1, 2, \dots, M\}$, we can construct a DSO with $O(n^{2.5794}M)$ preprocessing time and constant query time. With high probability over the randomized preprocessing algorithm, the DSO answers every possible query correctly.
\end{theorem}

\begin{remark}
	Our preprocessing algorithm uses fast \emph{rectangular} matrix multiplication algorithms. To express our time bound as a function of $\omega$, we could also simulate rectangular matrix multiplications by square matrix multiplications, e.g.~multiply an $n\times m$ matrix and an $m\times n$ matrix by $\lceil m/n\rceil$ square matrix multiplications of dimension $n$. In this case, the preprocessing time becomes $\tilde{O}(n^{2+1/(4-\omega)}M) < O(n^{2.6146}M)$.
\end{remark}

\begin{remark}[Comparison with Prior Works]
	The biggest advantage of our DSO is, of course, its fast preprocessing algorithm. In fact, the preprocessing time bound is only an $O(n^{0.051})$ factor away from the current best time bound for APSP. Our DSO is also the first one to break a barrier of $\tilde{\Omega}(n^{8/3})$ preprocessing time while keeping constant query time.\footnote{There are three previous DSOs with both sub-cubic preprocessing time and constant query time: \cite{GrandoniW12}, \cite{ChechikC20}, and \cite{Ren20}. (The query time of the first two DSOs can be brought down to constant using Observation 2.1 of \cite{Ren20}. In the case of \cite{GrandoniW12}, this increases the preprocessing time by an additive factor of $\tilde{O}(n^{3-\alpha})$.) Even when $\omega = 2$, the preprocessing time bounds of these DSOs are $\tilde{O}(n^{8/3})$ (setting $\alpha$ appropriately), $\tilde{O}(n^{14/5})$, and $\tilde{O}(n^{8/3})$ respectively.} However, our DSO has two drawbacks. First, it can only return the length of the shortest path. It does not suggest an efficient way to produce this path. Second, it does not support negative edge weights.
\end{remark}

We highlight two technical ingredients that are crucial for the preprocessing algorithm of our DSO.

\paragraph{Inverting a polynomial matrix modulo $x^r$.} Let $r$ be an integer parameter, and $\bfF$ be a polynomial matrix of degree $d$ (i.e.~each entry of $\bfF$ is a degree-$d$ polynomial over some formal variable $x$) that is invertible. We show how to compute $\bfF^{-1} \bmod x^r$ in time
\[\tilde{O}(dn^\omega)+(r^2/d)\cdot \MM(n, nd/r, nd/r)\cdot n^{o(1)}.\]
(That is, we only preserve the monomials in $\bfF^{-1}$ with degrees at most $r-1$.) Here, $\MM(n_1, n_2, n_3)$ is the time complexity of multiplying an $n_1\times n_2$ matrix and an $n_2\times n_3$ matrix.

It is shown in \cite{ZhouLS15} that we can compute the full $\bfF^{-1}$ (instead of $\bfF^{-1}\bmod x^r$) in $\tilde{O}(n^3d)$ time. We examine their algorithm carefully and adapt it to our case where we only want to compute $\bfF^{-1}\bmod x^r$. We modulo each polynomial in the intermediate steps of the algorithm by $x^r$, and use fast rectangular matrix multiplication to speed up the algorithm.

\begin{restatable}{theorem}{ThmInvertAlgo}\label{thm:invert-algo}
	Let $r$ be an integer, $\dbF$ be a finite field. Let $\bfF\in (\dbF[x]/\langle x^r\rangle)^{n\times n}$ be an $n\times n$ matrix over the ring of polynomials modulo $x^r$, and let $d\ge 1$ be an upper bound on the degrees of entries of $\bfF$. If $\bfF$ is invertible over $(\dbF[x]/\langle x^r\rangle)^{n\times n}$, the number of field operations to compute $\bfF^{-1}\bmod x^r$ is at most
	\[\tilde{O}(dn^\omega)+(r^2/d)\cdot\MM(n, nd/r, nd/r) \cdot n^{o(1)}.\]
\end{restatable}

\begin{remark}
	A square matrix $\bfF$ over the commutative ring $\caR$ is invertible if and only if $\det(\bfF)$ is a unit in $\caR$. In our case where $\caR = \dbF[x] / \langle x^r\rangle$, this is true if and only if the constant term of $\det(\bfF)$ is nonzero.
\end{remark}

\begin{remark}
	The idea of using polynomial matrices to capture distances is a common technique in graph algorithms. It has found many applications in static algorithms \cite{Sankowski05}, fault-tolerant algorithms \cite{vdBS19}, and dynamic algorithms \cite{Sankowski05-dynamic, BrandN19, BrandNS19}.
\end{remark}

\paragraph{Computing consistent shortest path trees.} Our DSO needs to invoke \cite[Observation 2.1]{Ren20} (see also \cite{BernsteinK09}), which needs a \emph{consistent} set of (incoming and outgoing) shortest path trees rooted at each vertex. Here, by \emph{consistent}, we mean that for every pair of vertices $u, v$ and any two shortest path trees $T_1$ and $T_2$ (from the $2n$ trees; recall they are \emph{directed} rooted trees), if $u$ can reach $v$ in both $T_1$ and $T_2$, then the $u\To v$ paths in $T_1$ and $T_2$ are the same path. In other words, we want to specify a \emph{unique} shortest path between each pair of vertices, such that for every vertex $v$, the shortest paths starting from $v$ (or ending at $v$, respectively) form a tree.

Note that this problem is quite nontrivial in small-weighted graphs. There may be many shortest paths between two vertices, and it is not obvious how to pick one shortest path for each vertex pair while guaranteeing consistency. Also, we cannot randomly perturb the edge weights by small values, as that would break the property that edge weights are small integers. It is also unclear how to construct such a set of shortest path trees from the APSP algorithm in \cite{Zwick02}. Previously, combining ideas in \cite[Section 3.4]{DemetrescuI04} and an algorithm in \cite{DuanP09}, \cite{Ren20} showed how to compute such shortest path trees in $\tilde{O}(n^{(3+\omega)/2}M) \le O(n^{2.6865}M)$ time; unfortunately, this time bound is worse than our claimed time bound $O(n^{2.5794}M)$ in \cref{thm:main}.

In this paper, we show how to construct consistent shortest paths trees in $O(n^{2.5286}M)$ time, matching the currently best time bound for APSP \cite{Zwick02}. Below is an informal statement, see \cref{thm:breaking-tie} for the precise version.

\begin{theorem}[Informal Version]\label{thm:unique-shortest-paths}
	Given a directed graph $G=(V, E)$ with edge weights in $\{1, 2, \dots, M\}$, we can compute a set of incoming and outgoing shortest path trees rooted at each vertex that are consistent, in $O(n^{2.5286}M)$ time.
\end{theorem}

\subsection{Warm-Up: DSO in $\tilde{O}(n^{(3+\omega)/2}M)$ Preprocessing Time}
Actually, the ideas in \cite{vdBS19} of maintaining the \emph{adjoint} of the \emph{symbolic adjacency matrix} (see \cref{sec:DSO}), together with ideas in \cite{Ren20}, already give us a DSO with $\tilde{O}(n^{(3+\omega)/2}M)$ preprocessing time and constant query time. As a warm-up, we briefly describe this DSO before we proceed into the details of \cref{thm:main}.

An \emph{$r$-truncated DSO} \cite{Ren20} is a DSO that only needs to be correct for the queries $(u, v, f)$ whose answer (i.e.~length of the corresponding shortest path) is at most $r$. If the answer is greater than $r$, it should return $r$ instead. In what follows, we will describe how to construct an $r$-truncated DSO in $\tilde{O}(rn^\omega)$ preprocessing time and $\tilde{O}(r)$ query time. Using techniques in \cite{Ren20} (see also \cref{sec:full-DSO}), this implies a DSO with $\tilde{O}(n^{(3+\omega)/2}M)$ preprocessing time and constant query time.

Let $\dbF$ be a sufficiently large finite field, and $\bfA$ be the following matrix. For every vertices $u, v$, if there is an edge from $u$ to $v$ with weight $l$, then let $\bfA_{u, v} = a_{u, v}x^l$, where $a_{u, v}$ is a random element in $\dbF$, and $x$ is an indeterminate. Furthermore, for every vertex $v$, let $\bfA_{v, v} = 1$. It is well-known \cite{Sankowski05} that with high probability over the choices of $a_{u, v}$, the \emph{adjoint} matrix of $\bfA$ encodes the shortest path information of the input graph, as follows. Let $\adj(\bfA)$ be the adjoint matrix of $\bfA$, and $u, v$ be two vertices, then the lowest degree of $\adj(\bfA)_{u, v}$ is exactly the distance from $u$ to $v$. For example, if $\adj(\bfA)_{u, v} = 7x^8 + 6x^5 - 9x^4$, then the distance from $u$ to $v$ is $4$.

A big advantage of the adjoint matrix, exploited in \cite{vdBS19} and also this work, is that it is easy to perform \emph{low-rank} updates, by the Sherman-Morrison-Woodbury formula (see \cref{thm:SMW-formula}). Given a matrix $\bfA$, its adjoint $\adj(\bfA)$, and a low-rank matrix $\bfB$, we can compute a specific element of $\adj(\bfA + \bfB)_{u, v}$, in time much faster than brute force. Therefore, we answer a query $(u, v, f)$ as follows: We first express the failure as a \emph{rank-one} matrix $\bfF$, such that $\bfA + \bfF$ is the matrix corresponding to the graph with $f$ removed. Then we can compute $\adj(\bfA + \bfF)_{u, v}$ quickly. Given this element (a polynomial over $\dbF$), we can easily compute the answer to the query.

What is the time complexity of this DSO? Recall that we only want to construct an $r$-truncated DSO, so we can modulo every entry in the process of computing $\adj(\bfA)$ by the polynomial $x^r$. Every arithmetic operation in the commutative ring $\dbF[x] / \langle x^r\rangle$ only takes $\tilde{O}(r)$ time. Computing the adjoint of a matrix reduces to inverting that matrix, which takes $\tilde{O}(n^\omega)$ arithmetic operations \cite{MatInv}. Therefore it takes $\tilde{O}(rn^\omega)$ time to compute $\adj(\bfA)\bmod x^r$. A close inspection of the Sherman-Morrison-Woodbury formula shows that each query can be completed in $O(1)$ arithmetic operations, i.e.~$\tilde{O}(r)$ time.

The $\tilde{O}(rn^\omega)$-time algorithm for inverting a polynomial matrix modulo $x^r$ is not optimal; the time bound in \cref{thm:invert-algo} is better. In \cref{sec:invert-poly-matrix}, we use fast rectangular matrix multiplication algorithms to speed up the algorithm in \cite{ZhouLS15}, obtaining a faster algorithm for inverting polynomial matrices modulo $x^r$.

\section{Preliminaries}\label{sec:preliminaries}

In this paper, we say an event happens \emph{with high probability} (w.h.p.) if it happens with probability at least $1-1/n^c$, for a constant $c$ that can be made arbitrarily large. Our DSOs (or $r$-truncated DSOs) will have a randomized preprocessing algorithm and a deterministic query algorithm. We say a DSO is \emph{correct with high probability} if w.h.p.~over its (randomized) preprocessing algorithm, it answers every possible query $(u, v, f)$ correctly.

\paragraph{Notation.} We use the following notation in \cite{DuanP09a, Ren20}.\begin{itemize}
	\item Let $p$ be a path, we use $|p|$ to denote the number of edges in $p$, and use $\|p\|$ to denote the length of $p$ (i.e.~total weight of edges in $p$).
	\item Let $u, v$ be two vertices, we define $\|uv\|$ as the length of the shortest path from $u$ to $v$. Furthermore, let $f$ be a failure (which is either an edge or a vertex), we define $\|uv\diamond f\|$ as the length of the shortest path from $u$ to $v$ that does not go through $f$.
	\item Let $u, v$ be two vertices, we define $|uv|$ as the number of edges in the shortest path from $u$ to $v$. In the case that there are many shortest paths from $u$ to $v$, it turns out that the following definition will be convenient in \cref{sec:breaking-tie}: We define $|uv|$ as the \emph{largest} number of edges in any shortest path from $u$ to $v$. %
\end{itemize}

\paragraph{Fast matrix multiplication.} Let $\omega$ be the exponent of matrix multiplication; the current best upper bound is $\omega < 2.3728596$ \cite{AlmanW21}. For positive integers $n_1, n_2, n_3$, let $\MM(n_1, n_2, n_3)$ denote the minimum number of arithmetic operations needed to multiply an $n_1\times n_2$ matrix and an $n_2\times n_3$ matrix. We define $\omega(a, b, c)$ to be the exponent of multiplying an $n^a\times n^b$ matrix and an $n^b\times n^c$ matrix, i.e.
\[\omega(a, b, c) = \inf\{w : \MM(n^a, n^b, n^c) = O(n^w)\}.\]
It is a classical result that $\omega(1, 1, \lambda) = \omega(1, \lambda, 1) = \omega(\lambda, 1, 1)$ for any real number $\lambda > 0$ \cite{LottiR83}; we denote $\omega(\lambda) = \omega(1, 1, \lambda)$.

We will need the following lemmas about the exponent of rectangular matrix multiplication. For completeness, we include proofs for these lemmas in \cref{sec:apd-FMM}.
\begin{restatable}{lemma}{MatMulI}\label{lemma:matmul1}
	Let $a, b, c, r$ be positive real numbers, then $r+\omega(a, b, c) \le \omega(a, b+r, c+r)$.
\end{restatable}

\begin{restatable}{lemma}{MatMulII}\label{lemma:matmul2}
	Consider the function $f(\tau) = \omega(1, 1-\tau, 1-\tau)$, where $\tau \in [0, 1]$. Then $\tau+f(\tau)$ is monotonically non-increasing in $\tau$, and $2\tau + f(\tau)$ is monotonically non-decreasing in $\tau$.
\end{restatable}

\paragraph{Polynomial operations.} Let $p, q\in\dbF[x]$ be two polynomials of degree $d$. It is easy to compute $p+q$ or $p-q$ in $O(d)$ field operations. We can also compute $p\cdot q$ in $\tilde{O}(d)$ field operations using fast Fourier transform. (Here, $\tilde{O}$ hides $\operatorname{polylog}(d)$ factors.) When $p$ is invertible, it is also possible to compute $p^{-1} \bmod x^d$ in $\tilde{O}(d)$ field operations \cite[Section 8.3]{AhoHU74}.

\section{Constructing a DSO in $O(n^{2.5794}M)$ Time}\label{sec:DSO}

In this section, we show how to preprocess a distance sensitivity oracle in $O(n^{2.5794}M)$ time, such that every query can be answered in constant time. Our preprocessing algorithm is randomized; with high probability over the preprocessing algorithm, the query algorithm always returns the correct answer.

\subsection{Preliminaries}

First, our preprocessing algorithm will use the following algorithm for inverting a polynomial matrix. A detailed description of this algorithm will be given in \cref{sec:invert-poly-matrix}.

\ThmInvertAlgo*

Let $G$ be a directed graph whose edge weights are integers in $[1, M]$. We define its \emph{symbolic adjacency matrix} $\SA(G)$ as (see \cite{Sankowski05})
\[\SA(G)_{i, j} = \begin{cases}
	1 & \text{if $i = j$},\\
	z_{i, j}x^{l} & \text{if there is an edge from $i$ to $j$ with weight $l$ in $G$},\\
	0 & \text{otherwise},
\end{cases}\]
where $z_{i, j}$ are unique variables corresponding to edges of $G$.

It will be inefficient to deal with these variables $z_{i, j}$, therefore we will pick a suitably large field $\dbF$, and substitute each variable $z_{i, j}$ by a random element in $\dbF$. However, we still keep the indeterminate $x$. Now, let $\bfZ$ be a matrix where each $\bfZ_{i, j} \in \dbF$, we will use $\SA_{\bfZ}(G)$ to denote the matrix $\SA(G)$ with each formal variable $z_{i, j}$ substituted by the field element $\bfZ_{i, j}$. Note that $\SA_{\bfZ}(G)$ is a polynomial matrix where every entry is a polynomial over $x$ with degree at most $M$.

We recall the definition of \emph{adjoint} matrix that will be crucial to our algorithm. Let $\bfA$ be an $n\times n$ matrix over a commutative ring $\caR$, and $i, j\in [n]$. We denote by $\bfA^{i, j}$ the matrix $\bfA$ with every element in the $i$-th row and the $j$-th column set to zero, except that $(\bfA^{i, j})_{i, j} = 1$. The adjoint matrix of $\bfA$, denoted as $\adj(\bfA)$, is an $n\times n$ matrix such that $\adj(\bfA)_{i, j} = \det(\bfA^{j, i})$ for every $i, j\in[n]$. A basic fact about $\adj(\bfA)$ is that if $\det(\bfA)$ is a unit of $\caR$, then $\adj(\bfA) = \det(\bfA) \cdot \bfA^{-1}$.

There is a close relationship between the distances in the graph $G$ and the entries in the adjoint of $\SA(G)$. Let $p$ be a multivariate polynomial, we define $\deg^*_x(p)$ as the lowest degree of the variable $x$ in any monomial of $p$. If $p=0$, then we define $\deg^*_x(p):=+\infty$. We have:

\begin{theorem}[{\cite[Lemma 4]{Sankowski05}}]\label{thm:Sankowski-adjoint}
	Let $G$ be a directed graph with positive integer weights, $i,j$ be two vertices. Then the distance from $i$ to $j$ in $G$ is $\deg^*_x(\adj(\SA(G))_{i, j})$.
\end{theorem}

We need the following theorem that allows us to maintain the adjoint of a matrix under \emph{rank-$1$} queries. (This theorem is a special case of \cite[Lemma 1.6]{vdBS19}.)

\begin{theorem}\label{thm:SMW-formula}
	Let $\caR$ be an arbitrary commutative ring, $\bfA\in \caR^{n\times n}$ be an invertible matrix, $\bfu, \bfv\in \caR^n$ be column vectors, and $\gamma = 1+\bfv^\sfT\bfA^{-1}\bfu$. Suppose $\gamma$ is invertible, then $\bfA+\bfu\bfv^\sfT$ is also invertible, and
	\[\adj(\bfA+\bfu\bfv^\sfT) = \det(\bfA)(\gamma\bfA^{-1} - (\bfA^{-1}\bfu\bfv^\sfT\bfA^{-1})).\]
	\vspace{-1.5em}
\end{theorem}
\begin{proof}[Proof Sketch]
	By the matrix determinant lemma, we have 
	\[\det(\bfA+\bfu\bfv^\sfT) = \gamma\cdot \det(\bfA).\]
	Since $\gamma$ is invertible, we can use the Sherman-Morrison-Woodbury formula \cite{ShermanM50, Woodbury50}:
	\[(\bfA+\bfu\bfv^\sfT)^{-1} = \bfA^{-1} - \gamma^{-1}(\bfA^{-1}\bfu\bfv^\sfT\bfA^{-1}).\]
	The theorem is proved by multiplying the above two formulas together.
\end{proof}

We need the Schwartz-Zippel lemma that guarantees the correctness of our randomized algorithm.

\begin{theorem}[{Schwartz-Zippel Lemma, \cite{Schwartz80, Zippel79}}]\label{thm:schwartz-zippel}
	Let $p(x_1,x_2,\dots,x_m)$ be a non-zero polynomial of (total) degree $d$ over a field $\dbF$. Let $S$ be a finite subset of $\dbF$, and $r_1,r_2,\dots,r_m$ be independently and uniformly sampled from $S$. Then
	\[\Pr[p(r_1,r_2,\dots,r_m) = 0] \le \frac{d}{|S|}.\]
\end{theorem}

We also need the following algorithm that computes the determinant of a polynomial matrix.

\begin{theorem}[\cite{Storjohann03,labahn2017fast}]\label{lem:det}
	Let $\bfB\in\dbF[x]^{n\times n}$ be a matrix of degree at most $d$, then we can compute $\det(\bfB)$ in $\tilde{O}(dn^\omega)$ field operations.
\end{theorem}

\subsection{Constructing an $r$-Truncated DSO}\label{sec:r-truncated-DSO}
Recall that for a failure $f$ (which is either a vertex or an edge), $\|uv\diamond f\|$ denotes the length of the shortest path from $u$ to $v$ that avoids $f$. An \emph{$r$-truncated} DSO, as defined in \cite{Ren20}, is a DSO that given a query $(u, v, f)$, outputs the value $\min\{\|uv\diamond f\|, r\}$. The main result of this subsection is that given an integer $r$ and an input graph $G$, an $r$-truncated DSO can be constructed in time
\[\tilde{O}(n^\omega M) + r^2/M\cdot \MM(n, nM/r, nM/r)\cdot n^{o(1)}.\]

\paragraph{Preprocessing algorithm.} Let $C$ be a large enough constant. First, we choose a prime $p \in [n^C, 2n^C]$ and let $\dbF=\dbZ_p$. Then we let $\bfZ$ be an $n\times n$ matrix over $\dbF$, where every $\bfZ_{i, j}$ is sampled independently from $\dbF$ uniformly at random. We substitute $\bfZ$ into $\SA(G)$ to obtain the matrix $\SA_{\bfZ}(G)$. Recall that each element of $\SA_{\bfZ}(G)$ is a polynomial over $x$ with coefficients in $\dbF$, whose degree is at most $M$. Then we compute $\SA_{\bfZ}(G)^{-1}$ and $\det(\SA_{\bfZ}(G))$ using \cref{thm:invert-algo} and \cref{lem:det} respectively.

Since we only want an $r$-truncated DSO, we only need to compute $\SA_{\bfZ}(G)^{-1}$ modulo $x^r$, i.e.~we only preserve the monomials with degree less than $r$ in every entry of $\SA_{\bfZ}(G)^{-1}$. Note that $\SA_{\bfZ}(G)$ is of the form $\bfI + x\bfM$ for some matrix $\bfM\in \dbF[x]^{n\times n}$, therefore its determinant is of the form $1 + x\cdot p(x)$ for some polynomial $p(x)$. As the determinant is invertible modulo $x^r$, $\SA_{\bfZ}(G)$ is also invertible modulo $x^r$. By \cref{thm:invert-algo}, we can compute $\SA_{\bfZ}(G)^{-1}\bmod x^r$ in time 
\[\tilde{O}(n^\omega M)+(r^2/M)\cdot \MM(n, nM/r, nM/r)\cdot n^{o(1)}.\]

By \cref{lem:det}, we can compute $\det(\SA_{\bfZ}(G))$ in $\tilde{O}(n^\omega M)$ time. Again, we only need to store the polynomial $\det(\SA_{\bfZ}(G)) \bmod x^r$. This concludes the preprocessing algorithm.

For the following query algorithms, we use $\bfe_i$ to denote the $i$-th standard unit vector, i.e.~$(\bfe_i)_i = 1$, and $(\bfe_i)_j = 0$ for every index $j \ne i$. %

\paragraph{Query algorithm for an edge failure.} A query consists of vertices $u, v\in V$ and a failed edge $e$. We assume that $e$ goes from vertex $a$ to vertex $b$, and has weight $l$. Let $G'$ be the graph obtained by removing $e$ from $G$, then we have $\SA(G') = \SA(G) + \bfu\bfv^\sfT$, where $\bfu=\bfe_a$ and $\bfv=-z_{a, b}x^l\bfe_b$. Let
\begin{itemize}
	\item $\gamma = 1 + \bfv^\sfT\SA(G)^{-1}\bfu = 1 - z_{a, b}x^l\SA(G)^{-1}_{b, a}$,
	\item $\beta = (\SA(G)^{-1}\bfu\bfv^{\sfT}\SA(G)^{-1})_{u, v} = -\SA(G)^{-1}_{u, a}z_{a, b}\SA(G)^{-1}_{b, v}x^l$, and
	\item $\alpha = \det(\SA(G))(\gamma\cdot\SA(G)^{-1}_{u, v} - \beta)$,
\end{itemize}
then by \cref{thm:SMW-formula}, we have $\alpha = \adj(\SA(G'))_{u, v}$. (Note that since $l \ge 1$, $\gamma$ is always invertible.) %

\paragraph{Query algorithm for a vertex failure.} A query consists of vertices $u,v\in V$ and a failed vertex $f\in V$. It suffices to remove every outgoing edge from $f$ (and we do not need to also remove incoming edges to $f$), as $f$ already cannot appear as an intermediate vertex in every path from $u$ to $v$. Therefore, we need to compute $\adj(\SA(G'))_{u,v}$, where $G'$ is obtained by removing all outgoing edges from $f$ in $G$. Let $\bfu = \bfe_f$, and $\bfv$ be the negation of the transpose of the $f$-th row of $\SA(G)$, except that $\bfv_f = 0$, i.e.,  
\[\bfv_j = \begin{cases}
	-z_{f, j}x^{l} & \text{if there is an edge from $f$ to $j$ with weight $l$ in $G$},\\
	0 & \text{otherwise},
\end{cases}\] It is easy to see $\SA(G')=\SA(G)+\bfu\bfv^\sfT$. To compute $\adj(\SA(G'))_{u,v}$ using \cref{thm:SMW-formula}, we let
\begin{itemize}
    \item $\gamma=1+\bfv^\sfT\SA(G)^{-1}\bfu$. Note that $(\bfe_f - \bfv)^\sfT$ is exactly the $f$-th row of $\SA(G)$, so $(\bfe_f - \bfv)^\sfT\SA(G)^{-1} = \bfe_f^\sfT$, and $\bfv^\sfT\SA(G)^{-1} = \bfe_f^\sfT\SA(G)^{-1} - \bfe_f^\sfT$. We have $\gamma = 1+\bfe_f^\sfT\SA(G)^{-1}\bfu - \bfe_f^\sfT\bfu = \SA(G)^{-1}_{f,f}$;%
    \item $\beta = (\SA(G)^{-1}\bfu\bfv^\sfT\SA(G)^{-1})_{u,v}=(\bfe^\sfT_u\SA(G)^{-1}\bfu)(\bfv^\sfT\SA(G)^{-1}\bfe_v)=\SA(G)^{-1}_{u,f}(\bfe_f^\sfT\SA(G)^{-1}\bfe_v) = \SA(G)^{-1}_{u,f}\SA(G)^{-1}_{f,v}$;
    \item and $\alpha = \det(\SA(G))(\gamma\cdot\SA(G)^{-1}_{u,v}-\beta)$,
\end{itemize} then we have $\alpha = \adj(\SA(G'))_{u, v}$. (Note that $\gamma$ is always invertible since the constant term of $\SA(G)^{-1}_{f,f}$ must be $1$.)

In the actual query algorithm, we will substitute each formal variable $z_{i, j}$ by $\bfZ_{i, j}$. Let $\gamma_\bfZ$ denote the resulting polynomial after this substitution. Note that $\gamma_\bfZ$ is a polynomial in $\dbF[x]$. Similarly we can define $\beta_\bfZ$ and $\alpha_\bfZ$. If $\alpha_\bfZ \not\equiv 0\pmod{x^r}$, then our query algorithm outputs $\deg_x^*(\alpha_\bfZ)$; otherwise it outputs $r$.

From the above formulas, we can compute $\gamma_\bfZ$, $\beta_\bfZ$, and $\alpha_\bfZ$ in $O(1)$ arithmetic operations over polynomials. Note that we only need to compute these polynomials modulo $x^r$, so each such arithmetic operation takes $\tilde{O}(r)$ time. The total query time is thus $\tilde{O}(r)$.

\begin{remark}[Query Algorithm for Undirected Graphs]
	Our $r$-truncated DSO can also deal with undirected graphs, but the details are a bit different from the case of directed graphs. To remove an undirected edge, we need to update two entries in $\SA(G)$, which corresponds to a rank-$2$ update to $\SA(G)$. To remove a vertex, we need to update one row and one column in $\SA(G)$, which is also a rank-$2$ update to $\SA(G)$. Therefore, we need to use the rank-$2$ version of \cref{thm:SMW-formula} (see \cite[Lemma 1.6]{vdBS19}). Actually, our $r$-truncated DSOs also support deleting $f$ failures, and the query time is $\tilde{O}(f^\omega r)$. We omit the details here and refer the interested readers to \cite{vdBS19}.
\end{remark}

\begin{theorem}\label{thm:r-trunc}
	For every integer $r$, we can construct an $r$-truncated DSO with preprocessing time
	\[\tilde{O}(n^\omega M)+r^2/M\cdot \MM(n,nM/r,nM/r)\cdot n^{o(1)},\]
	and query time $\tilde{O}(r)$. Our $r$-truncated DSO is correct w.h.p.
\end{theorem}

(Recall that by saying our $r$-truncated DSO is correct w.h.p, we mean that w.h.p.~over its randomized preprocessing algorithm, it answers every query correctly.)

\begin{proof}[Proof of \cref{thm:r-trunc}]
	We only need to prove the correctness of our $r$-truncated DSO. Consider a query $(u,v,f)$ where $f$ is an edge or a vertex, and let $G'$ be the graph obtained by removing $f$ from $G$. By \cref{thm:SMW-formula}, we have $\alpha_\bfZ = \adj(\SA_\bfZ(G'))_{u, v}$. (Note that the constant term of $\gamma_\bfZ$ is always $1$, so $\gamma_\bfZ$ is always invertible.)

	If $\|uv\diamond f\|\ge r$, then by \cref{thm:Sankowski-adjoint}, $\adj(\SA(G'))_{u, v}$ must be a polynomial whose minimum degree over $x$ is at least $r$. In this case, we have $\alpha_\bfZ \equiv 0\pmod {x^r}$ for every $\bfZ$. Therefore, our algorithm returns $r$, which is correct.

	If $\|uv\diamond f\|=k<r$, then by \cref{thm:Sankowski-adjoint}, $\adj(\SA(G'))_{u,v}$ must be a polynomial whose minimum degree is exactly $k$. In this case, the coefficient of $x^k$ in $\alpha$ is a polynomial of $z_{i, j}$ with (total) degree at most $n$. (This is because $\adj(\SA(G'))_{u, v}$ is the determinant of a certain $n\times n$ matrix in which every entry has total degree at most one in the variables $z_{i, j}$.) If this polynomial is nonzero at $\bfZ$, then $\deg_x^*(\alpha_\bfZ) = k$ and our query algorithm is correct. By \cref{thm:schwartz-zippel}, this polynomial is $0$ with probability at most $1/n^{C-1}$. Therefore, our query algorithm returns the correct answer $k$ with probability at least $1-1/n^{C-1}$.
	
	In conclusion, for every fixed query $(u, v, f)$, our query algorithm is correct with probability $1-1/n^{C-1}$ over the choice of $\bfZ$. By a union bound over $O(n^4)$ possible queries, the probability (over our randomized preprocessing algorithm) that every query is answered correctly is at least $1-1/\Theta(n^{C-5})$, which is a high probability.
\end{proof}

\subsection{Constructing the Full DSO}\label{sec:full-DSO}
Now we have constructed an $r$-truncated DSO, which we denote by $\Dstart$. In this subsection, we will extend it to a \emph{full} DSO using the techniques in \cite{Ren20}. Specifically, we use the following two algorithms from \cite{Ren20}.

The first algorithm transforms an ($r$-truncated) DSO with a possibly large query time into an ($r$-truncated) DSO with query time $O(1)$. More precisely:
\begin{lemma}[{\cite[Observation 2.1]{Ren20}}]\label{lem:fast}
	Given an $r$-truncated DSO $\mathcal{D}$ with preprocessing time $P$ and query time $Q$, we can build an $r$-truncated DSO $\Fast(\mathcal{D})$ with query time $O(1)$ which is correct w.h.p. The preprocessing algorithm of $\Fast(\caD)$ is as follows:
	\begin{itemize}
	    \item It needs the all-pairs distance matrix of the input graph $G$, as well as the set of \emph{consistent} (incoming and outgoing) shortest path trees rooted at each vertex in $G$. By \cref{thm:unique-shortest-paths}, these shortest path trees can be computed in $O(n^{2.5286}M)$ time. For details, see \cref{sec:breaking-tie}.
	    \item It invokes the preprocessing algorithm of $\mathcal{D}$ on the input graph $G$ once, and makes $\tilde{O}(n^2)$ queries to $\mathcal{D}$. The preprocessing time is $P+\tilde{O}(n^2)Q$.
	\end{itemize}
\end{lemma}

The second algorithm we use is implicit in the argument of \cite[Section 2.3]{Ren20}. We formalize it as the following lemma.
\begin{lemma}\label{lem:extend}
	Given an $r$-truncated DSO $\mathcal{D}$ with preprocessing time $P$ and query time $O(1)$, we can build a $(3/2)r$-truncated DSO $\Extend(\mathcal{D})$ with preprocessing time $P+O(n^2)$ and query time $\tilde{O}(nM/r)$. The new DSO is correct w.h.p.
\end{lemma}

Now, we are ready to explain our algorithm to build a full DSO. Given an $r$-truncated DSO $\Dstart$, we first obtain an $r$-truncated DSO $\caD_0$ with query time $O(1)$ by applying \cref{lem:fast}.

\def\istar{i^{\star}}
Let $\istar = \lfloor\log_{3/2}(nM/r)\rfloor$. For every $0\le i\le \istar$, we construct an $r(3/2)^{i+1}$-truncated DSO $\mathcal{D}_{i+1}$ by applying \cref{lem:extend} and \cref{lem:fast} sequentially on $\mathcal{D}_i$, i.e.~$\mathcal{D}_{i+1}=\Fast(\Extend(\mathcal{D}_i))$. Let the resulting DSO be $\Dfinal=\mathcal{D}_{\istar + 1}$, since $r(3/2)^{\istar + 1} \ge nM$, $\Dfinal$ is a full DSO.

We can also summarize our construction algorithm in one formula:
\[\Dfinal=\underbrace{\Fast(\Extend(\Fast(\Extend(\cdots \Fast(\Dstart)))))}_{O(\log (nM/r)) \text{ times}}.\]

\paragraph{Complexity of our DSO.} 
Let $r=Mn^\alpha$, where $\alpha\in[0,1]$ is a parameter to be determined.
By \cref{thm:r-trunc}, the preprocessing time of $\Dstart$ is
\[\tilde{O}(n^\omega M)+r^2/M\cdot \MM(n,nM/r,nM/r)\cdot n^{o(1)}\le \tilde{O}(n^\omega M) + n^{2\alpha+\omega(1,1-\alpha,1-\alpha) + o(1)}M,\]
and the query time of $\Dstart$ is $\tilde{O}(r) = \tilde{O}(n^\alpha M)$. By \cref{lem:fast}, the preprocessing time of $\mathcal{D}_0$ is
\[\tilde{O}(n^{2+\alpha}M + n^\omega M)+n^{2\alpha+\omega(1,1-\alpha,1-\alpha) + o(1)}M.\]

Now consider the preprocessing algorithm of $\Dfinal$. We need to compute the all-pairs distance matrix and in/out shortest path trees of $G$ as required by \cref{lem:fast}, which takes $\tilde{O}(n^{2+\mu}M)$ time by \cref{thm:unique-shortest-paths}. We also need to run the preprocessing algorithm of $\caD_0$. Also, for every $0\le i\le \istar$, we need to preprocess the oracle $\caD_{i+1}$, which takes $n^2\cdot \tilde{O}(nM / (r(3/2)^{i+1})) = \tilde{O}\mleft(\frac{n^{3-\alpha}M}{(3/2)^i}\mright)$ time.

Therefore, the preprocessing time of $\Dfinal$ is:
\begin{align*}
&\,\tilde{O}(n^{2 + \alpha}M + n^\omega M + n^{2+\mu}M)+n^{2\alpha + \omega(1,1-\alpha,1-\alpha) + o(1)}M+\sum_{i=0}^{\lfloor \log_{3/2} (nM/r)\rfloor} \tilde{O}\mleft(\frac{n^{3-\alpha}M}{(3/2)^i}\mright)\\
\le &\, n^{\max\{2+\alpha, 2+\mu, 3-\alpha, 2\alpha + \omega(1, 1-\alpha, 1-\alpha)\} + o(1)}M.
\end{align*}

Let $\alpha = 0.420645$, $\beta = \frac{1}{1-\alpha}$, then $1.5 < \beta < 1.75$. Recall that for any real number $\lambda$, $\omega(\lambda)$ is a shorthand for $\omega(1, 1, \lambda)$. We have
\begin{align}
	\omega(1, 1-\alpha, 1-\alpha) =&\, (1-\alpha)\omega(\beta)\nonumber\\
	\le&\,(1-\alpha)\cdot\frac{(1.75-\beta)\omega(1.5) + (\beta - 1.5)\omega(1.75)}{1.75 - 1.5}\label{eq:exponent-step2}\\
	\le&\,0.579355 \cdot 4\cdot (0.023943\cdot \omega(1.5) + 0.226058\cdot\omega(1.75))\nonumber\\
	\le&\,1.738094.\label{eq:exponent-step3}
\end{align}

Here, \cref{eq:exponent-step2} uses the convexity of the $\omega(\cdot)$ function \cite{LottiR83}, and \cref{eq:exponent-step3} uses the recent bounds in \cite{GallU18} that $\omega(1.5) \le 2.796537$ and $\omega(1.75) \le 3.021591$. We can see that 
\[\max\{2+\alpha, 2+\mu, 3-\alpha, 2\alpha + \omega(1, 1-\alpha, 1-\alpha)\} = 2\alpha + \omega(1, 1-\alpha, 1-\alpha) \le 2.579384.\]

By \cref{lem:fast}, the query time of $\Dfinal$ is $O(1)$. Therefore, we can construct a DSO with $O(n^{2.5794}M)$ preprocessing time and $O(1)$ query time.

As the DSOs constructed in \cref{lem:fast} always have size $\tilde{O}(n^2)$, our final DSO only occupies $\tilde{O}(n^2)$ space. However, we remark that the preprocessing algorithm of our DSO requires $\tilde{O}(rn^2) = O(n^{2.4207})$ space (in particular, to store $\SA_{\bfZ}(G)^{-1}\bmod x^r$).

\section{Inverting a Polynomial Matrix Modulo $x^r$}\label{sec:invert-poly-matrix}

As we see in \cref{sec:DSO}, the algorithm in \cref{thm:invert-algo} for inverting a polynomial matrix modulo $x^r$ is very crucial for our results.

\ThmInvertAlgo*

In this section, we work in a (large enough) field $\dbF$, and regard each polynomial in the matrix as an element of the commutative ring $\caR = \dbF[x] / \langle x^r\rangle$. Without loss of generality, we assume $n$ and $r$ are powers of $2$ throughout this section.

\subsection{An Informal Treatment}\label{sec:informal-invert-algo}
Our algorithm is essentially the algorithm in \cite{ZhouLS15}. In fact, the only difference is that we only consider polynomials modulo $x^r$. In \cref{sec:proof-invert-algo}, we will provide an improved analysis of this algorithm by using rectangular matrix multiplication. Here we present a brief exposition of the algorithm in \cite{ZhouLS15}.

Let $\bfF$ be an input polynomial matrix where each entry has degree at most $d$. Suppose $\bfF$ is invertible over $(\dbF[x]/\langle x^r\rangle)^{n\times n}$. We will compute a \emph{kernel basis decomposition} of $\bfF$, which is a chain of matrices $\bfA_1, \bfA_2, \dots, \bfA_{\log n}$ and a diagonal matrix $\bfB$, such that
\begin{equation}
\bfF^{-1} = \bfA_1\bfA_2\dots\bfA_{\log n}\bfB^{-1}.\label{eq:kernel-basis-decomp}
\end{equation}

Then, to compute $\bfF^{-1}$, we simply multiply the above matrices. Note that $\bfB$ is a \emph{diagonal} matrix that is invertible\footnote{Every diagonal element of $\bfB$ is a divisor of the \emph{largest invariant factor} of $\bfF$ (see \cite[Section 5.1]{ZhouLS15}), which is (again) a divisor of $\det(\bfF)$. Since $\det(\bfF)$ is invertible modulo $x^r$, every diagonal element of $\bfB$ is also invertible modulo $x^r$.}, so its inverse is easy to compute.

To start, we write $\bfF = \begin{bmatrix}\bfF_{\sfU}\\\bfF_{\sfD}\end{bmatrix}$, where each $\bfF_\sfU$ or $\bfF_\sfD$ is an $(n/2)\times n$ matrix. Then we compute two $n\times (n/2)$ matrices $\bfN_\sfR$ and $\bfN_\sfL$ with full rank, such that $\bfF_\sfU\bfN_\sfR = {\bf 0}$, and $\bfF_\sfD\bfN_\sfL = {\bf 0}$. (This can be done by \cite[Theorem 4.2]{ZhouLS12}.) Let $\bfA_1 = \begin{bmatrix}\bfN_\sfL & \bfN_\sfR\end{bmatrix}$, then $\bfA_1$ has full rank, and
\[\bfF\cdot \bfA_1 = \begin{bmatrix}\bfF_\sfU\bfN_\sfL & \bfF_\sfU\bfN_\sfR \\ \bfF_\sfD\bfN_\sfL & \bfF_\sfD\bfN_\sfR\end{bmatrix} = \begin{bmatrix}\bfF_\sfU\bfN_\sfL & \\ & \bfF_\sfD\bfN_\sfR\end{bmatrix}.\]

Therefore, $\bfF\cdot \bfA_1$ is a block diagonal matrix with two blocks, each of size $(n/2) \times (n/2)$. We can then recursively invoke the kernel basis decomposition of these two blocks, and form the matrices $\bfA_2, \dots, \bfA_{\log n}$. The diagonal matrix $\bfB$ is created at the base case of the recursion, where the diagonal blocks of $\bfF\cdot \bfA_1\cdot\dots\cdot \bfA_{\log n}$ are of size $1\times 1$. It is shown in \cite{ZhouLS15} that the kernel basis decomposition takes only $\tilde{O}(dn^\omega)$ time to compute.

We still need to compute \cref{eq:kernel-basis-decomp}. From the above algorithm, we can see that each $\bfA_i$ is a block-diagonal matrix, which consists of $2^{i-1}$ blocks of size $(n/2^{i-1})\times (n/2^{i-1})$. Now we \emph{assume} that each entry in $\bfA_i$ also has degree at most $d\cdot 2^{i-1}$. (In reality, the behavior of degrees in $\bfA_i$ may be complicated, and we need the notion of \emph{shifted column degree} (see \cref{def:shifted-column-degree}) to control it.)

To compute \cref{eq:kernel-basis-decomp}, we define $\bfM_i = \bfA_1\bfA_2\dots \bfA_i$, and compute each $\bfM_i$ by the formula
\begin{equation}
	\bfM_{i+1} = \bfM_i\bfA_{i+1}.
	\label{eq:Mi+1}
\end{equation}

The degree of each entry in $\bfM_i$ will be at most $O(2^i\cdot d)$. As we only need the results modulo $x^r$, we can assume the degrees are actually $O\mleft(\min\{r, 2^i\cdot d\}\mright)$. Note that $\bfA_{i+1}$ consists of $2^i$ blocks, each of size $(n/2^i)\times (n/2^i)$, and the degree of each (nonempty) entry in $\bfA_{i+1}$ is also $O\mleft(\min\{r, 2^i\cdot d\}\mright)$. Therefore, we can compute \cref{eq:Mi+1} in
\begin{equation}
O\mleft(\min\{r, 2^i\cdot d\}\mright)\cdot 2^i\cdot \MM(n, n/2^i, n/2^i)
\label{eq:time-for-Mi+1}
\end{equation}
time. (It is basically $2^i$ matrix products of size $n\times (n/2^i)$ and $(n/2^i)\times (n/2^i)$; we need to multiply another factor of $\min\{r, 2^i\cdot d\}$ which is the degree of polynomials in these matrices.)

Now, it is easy to see that the bottleneck of this algorithm occurs when $r = 2^i\cdot d$, and the time for computing \cref{eq:Mi+1} is:
\[(\text{\ref{eq:time-for-Mi+1}}) = (r^2/d)\cdot \MM(n, nd/r, nd/r).\]

\subsection{Proof of \cref{thm:invert-algo}}\label{sec:proof-invert-algo}

As opposed to the informal description above, the \emph{maximum} degrees in the matrices may not behave well. We need to introduce the concept of \emph{column degrees} and \emph{shifted column degrees} to capture the behavior of the degrees in these matrices.

\begin{definition}[{\cite[Section 2.2]{ZhouLS15}}]\label{def:shifted-column-degree}
	Let $\vec{\bfp}$ be a length-$n$ column vector whose entries are polynomials. Then the \emph{column degree} of $\vec{\bfp}$, denoted as $\cdeg\vec{\bfp}$, is the maximum of the degrees of the entries in $\vec{\bfp}$. That is:
	\[\cdeg\vec{\bfp} = \max_{i=1}^n\{\deg(\bfp_i)\}.\]
	
	Let $\vec{s}$ be a length-$n$ vector of integers, called the \emph{shift} of the degrees. Then the \emph{$\vec{s}$-shifted column degree} of $\vec{\bfp}$, or simply the \emph{$\vec{s}$-column degree} of $\vec{\bfp}$, denoted as $\cdeg_{\vec{s}}\vec{\bfp}$, is defined as
	\[\cdeg_{\vec{s}}\vec{\bfp} = \max_{i=1}^n\{s_i + \deg(\bfp_i)\}.\]
	
	It is easy to see that $\cdeg\vec{\bfp} = \cdeg_{\vec{\bf 0}}\vec{\bfp}$, where $\vec{\bf 0}$ is the all-zero vector.
	
	Let $\bfA$ be an $m\times n$ polynomial matrix, then the \emph{column degree} (\emph{$\vec{s}$-column degree} resp.) of $\bfA$, denoted as $\cdeg\bfA$ ($\cdeg_{\vec{s}}\bfA$ resp.), is the length-$n$ row vector whose $i$-th entry is the column degree ($\vec{s}$-column degree resp.) of the $i$-th column of $\bfA$.
\end{definition}

We need the following theorem. It is essentially Theorem 3.7 of \cite{ZhouLS12}, where we replace the invocations of square matrix multiplication algorithms with (the faster) rectangular matrix multiplication algorithms. It is straightforward to adapt the original proof in \cite{ZhouLS12} to use rectangular matrix multiplication, but for completeness, we will include a proof in \cref{sec:unbalanced-mat-mul}.

\begin{restatable}{theorem}{ThmUnbalancedMatMul}\label{thm:unbalanced-mat-mul}
	Let $\bfA$ be an $n^p\times n^q$ polynomial matrix, and $\bfB$ be an $n^q\times n^r$ polynomial matrix. Suppose $\vec{s} \ge \cdeg\bfA$ is a shift that bounds the corresponding column degrees of $\bfA$, and 
	\[\xi = \max\mleft\{\frac{1}{n^q}\sum_{i=1}^{n^q} s_i, \frac{1}{n^r}\sum_{i=1}^{n^r} (\cdeg_{\vec{s}}\bfB)_i\mright\} + 1.\]
	Then the product $\bfA \cdot \bfB$ can be computed in $\xi\cdot n^{\omega(p, q, r) + o(1)}$ field operations.
\end{restatable}

Now we can prove \cref{thm:invert-algo}.

\ThmInvertAlgo*

\begin{proof}[Proof Sketch]
	In this sketch, we will use some results in \cite{ZhouLS15} directly. We will also use some notation introduced in \cref{sec:informal-invert-algo}.
	
	Let $\vec{s} = \cdeg\bfF$. We first invoke the kernel basis decomposition algorithm \textsc{Inverse} of \cite{ZhouLS15}:
	\[(\bfA_1, \bfA_2, \dots, \bfA_{\log n}, \bfB) \gets \textsc{Inverse}(\bfF, \vec{s}).\]

	By \cite[Theorem 8]{ZhouLS15}, the algorithm \textsc{Inverse} takes only $\tilde{O}(dn^\omega)$ time.Then we compute
	\[\bfF^{-1} = \bfA_1\bfA_2\dots\bfA_{\log n}\bfB^{-1}.\]
	
	Note that $\bfB$ is a diagonal matrix, so it suffices to compute $\bfA_1\bfA_2\dots \bfA_{\log n}$. Also recall that for every $0\le i < \log n$, $\bfA_{i+1}$ is a block diagonal matrix that consists of $2^i$ diagonal blocks of size $(n/2^i)\times (n/2^i)$. Let $\bfA^{(j)}_{i+1}$ denote the $j$-th block, we write
	\[\bfA_{i+1} = \diag(\bfA^{(1)}_{i+1}, \dots, \bfA^{(2^i)}_{i+1}).\]
	
	Let $\bfM_i = \bfA_1\bfA_2\dots \bfA_i$. Then for every $1\le i < \log n$,
	\begin{equation}
	\bfM_{i+1} = \bfM_i\bfA_{i+1}.
	\tag{\ref{eq:Mi+1}}
	\end{equation}
	
	In order to use results in \cite[Lemma 10]{ZhouLS15}, we need to partition each $\bfA_{i+1}^{(k)}$ into two kernel bases. Like how $\bfA_1$ was formed in \cref{sec:informal-invert-algo}, we denote $\bfA_{i+1}^{(k)} = \begin{bmatrix}\bfN^{(k)}_{i+1, \sfL} & \bfN^{(k)}_{i+1, \sfR}\end{bmatrix}$. Here, each $\bfN^{(k)}_{i+1, \sfL}$ or $\bfN^{(k)}_{i+1, \sfR}$ is of dimension $(n/2^i) \times (n/2^{i+1})$. We divide $\bfM_i$ into submatrices (``column blocks'') of dimension $n\times (n/2^i)$ accordingly:
	\[
	\bfM_i = \begin{bmatrix}
	\bfM_i^{(1)} & \bfM_i^{(2)} & \dots & \bfM_i^{(2^i)}
	\end{bmatrix}.
	\]
	Then \cref{eq:Mi+1} is equivalent to
	\begin{equation}\label{eq:Mi+1-equiv}
	\bfM^{(2k-1)}_{i+1} = \bfM^{(k)}_i \cdot \bfN^{(k)}_{i+1, \sfL}, \text{ and }\bfM^{(2k)}_{i+1} = \bfM^{(k)}_i \cdot \bfN^{(k)}_{i+1, \sfR}.
	\end{equation}
	
	We use \cref{thm:unbalanced-mat-mul} to multiply these matrices. For each $1\le i < \log n$, in \cref{eq:Mi+1-equiv}, we need to perform $2^{i+1}$ matrix multiplications of the form $\bfM \cdot \bfN$. Here $\bfM=\bfM^{(k)}_i$, and $\bfN$ is either $\bfN^{(k)}_{i+1, \sfL}$ or $\bfN^{(k)}_{i+1, \sfR}$. The dimension of $\bfM$ is $n\times (n/2^i)$, and the dimension of $\bfN$ is $(n/2^i)\times (n/2^{i+1})$. Moreover, let $\vec{t} = \cdeg_{\vec{s}}\bfM_i^{(k)}$, then by \cite[Lemma 10]{ZhouLS15}:
	
	\begin{enumerate}[(a)]
		\item $\sum_{j=1}^{n/2^i}t_j \le \sum_{j=1}^n s_j \le dn$.
		\item $\sum_{j=1}^{n/2^{i+1}}(\cdeg_{\vec{t}}\bfN^{(k)}_{i+1, \sfL})_j \le \sum_{j=1}^n s_j \le dn$; similarly, $\sum_{j=1}^{n/2^{i+1}}(\cdeg_{\vec{t}}\bfN^{(k)}_{i+1, \sfR})_j \le dn$.
	\end{enumerate}
	(Recall that $\vec{s}$ is the column degree of $\bfF$.)
	
	Let
	\[\xi_i = \max\mleft\{\frac{1}{n/2^i}\sum_{j=1}^{n/2^i}t_j, \frac{1}{n/2^{i+1}}\sum_{k=1}^{n/2^{i+1}}(\cdeg_{\vec{t}}\bfN)_k\mright\} \le 2^{i+1}\cdot d.\]
	Note that we are only interested in the polynomials modulo $x^r$, thus by definition, every element in $\vec{t}$ and $\cdeg_{\vec{t}}\bfN$ should be upper bounded by $O(r)$. Therefore if $2^{i+1}d \ge r$, we use the bound $\xi_i \le O(r)$ instead. By \cref{thm:unbalanced-mat-mul}, the time complexity for computing $\bfM \cdot \bfN$ is $\xi_i \cdot n^{\omega(1, 1-\tau, 1-\tau) + o(1)}$, where $\tau = \log_n (2^{i+1})$.
	
	Let $\tau^\star = \frac{\log(r/d)}{\log n}$ be the threshold such that $2^{i+1}d\le r$ if and only if $\tau \le \tau^\star$. Suppose $2^{i+1}d \le r$, then the time complexity for computing all $2^{i+1}$ ($=n^\tau$) matrix products is
	\begin{align*}
		&~n^\tau\cdot \xi_i \cdot n^{\omega(1, 1-\tau, 1-\tau) + o(1)} \\
		\le&~d\cdot n^{2\tau + \omega(1, 1-\tau, 1-\tau) + o(1)}\\
		\le&~d\cdot n^{2\tau^\star + \omega(1, 1-\tau^\star, 1-\tau^\star) + o(1)} & \text{By \cref{lemma:matmul2}}\\
		\le&~(r^2/d)\cdot \MM(n, nd/r, nd/r) \cdot n^{o(1)}.
	\end{align*}
	On the other hand, suppose $2^{i+1}d > r$, then the time complexity for computing all $n^\tau$ matrix products is
	\begin{align*}
		&~n^\tau\cdot r\cdot n^{\omega(1, 1-\tau, 1-\tau) + o(1)}\\
		\le&~r\cdot n^{\tau^\star + \omega(1, 1-\tau^\star, 1-\tau^\star) + o(1)} & \text{By \cref{lemma:matmul2}}\\
		\le&~(r^2/d)\cdot \MM(n, nd/r, nd/r) \cdot n^{o(1)}.
	\end{align*}
	
	Summing over every $1\le i < \log n$, we can see that the time complexity for inverting $\bfF$ is at most
	\[\tilde{O}(dn^\omega)+(r^2/d)\cdot \MM(n, nd/r, nd/r) \cdot n^{o(1)}. \qedhere\]
\end{proof}

\subsection{Proof of \cref{thm:unbalanced-mat-mul}}\label{sec:unbalanced-mat-mul}

\ThmUnbalancedMatMul*

\begin{proof}
	W.l.o.g.~we assume that $n^p, n^q, n^r$ are powers of $2$. For every $1\le c\le r\log n-1$, let $\bfB^c$ denote the set of columns of $\bfB$ whose $\vec{s}$-column degrees are in the range $(2^c\xi, 2^{c+1}\xi]$; let $\bfB^0$ denote the rest columns of $\bfB$, i.e.~those with $\vec{s}$-column degrees no more than $2\xi$. Then $\bfB^0, \bfB^1, \dots, \bfB^{r\log n-1}$ form a partition of the columns of $\bfB$. By the definition of $\xi$, for every $0\le c\le r\log n-1$, there are at most $n^r/2^c$ columns in $\bfB^c$. To compute $\bfA\cdot \bfB$, it suffices to compute $\bfA\cdot \bfB^c$ for each $c$.
	
	Now fix an integer $c$, we need to compute $\bfA \cdot \bfB^c$. Using the same method above, we can also partition the columns of $\bfA$ into $q\log n$ groups. More precisely, for every $1\le c' \le q\log n-1$, let $\bfA^{c'}$ be the set of columns of $\bfA$ whose column degrees are in the range $(2^{c'}\xi, 2^{c'+1}\xi]$; let $\bfA^0$ be the rest columns of $\bfA$, i.e.~those with column degrees no more than $2\xi$. For notational convenience, we may assume that
	\[\bfA = \begin{bmatrix}\bfA^0 & \bfA^1 & \ldots & \bfA^{q\log n - 1}\end{bmatrix},\]
	as otherwise we can rearrange the columns of $\bfA$ (along with the rows of $\bfB$ and the entries in $\vec{s}$). We also note that for every $0 \le c' \le q\log n - 1$, there are at most $n^q / 2^{c'}$ columns in $\bfA^{c'}$.
	
	The partition of columns of $\bfA$ induces a partition of rows of $\bfB^c$. In particular, we define $\bfB^{c, c'}$ as the rows of $\bfB^c$ corresponding to columns of $\bfA^{c'}$, so
	\[\bfB^c = \begin{bmatrix}\bfB^{c, 0}\\ \bfB^{c, 1}\\ \vdots\\ \bfB^{c, q\log n-1}\end{bmatrix}.\]
	
	We can see that for every $c' > c$, $\bfB^{c, c'}$ is the zero matrix. In fact, suppose the entry in the $j$-th row and $k$-th column of $\bfB^c$ is nonzero, and this entry belongs to $\bfB^{c, c'}$ for some $c' > c$. Denote this column as $\bfb_k$, then $\cdeg_{\vec{s}}\bfb_k \ge s_j$. As the $j$-th column of $\bfA$ belongs to $\bfA^{c'}$, we have $s_j > 2^{c'}\xi \ge 2^{c+1}\xi$. However, by definition of $\bfB^c$, we also have $\cdeg_{\vec{s}}\bfb_k \le 2^{c+1}\xi$, a contradiction. Therefore
	\[\bfA\cdot \bfB^c = \sum_{c' = 0}^c\bfA^{c'} \cdot \bfB^{c, c'}.\]
	
	Again, fix $c' \in [0, c]$, we want to compute $\bfA^{c'} \cdot \bfB^{c, c'}$. Recall that the dimension of $\bfA^{c'}$ is at most $n^p\times (n^q/2^{c'})$, and each entry in $\bfA^{c'}$ is a polynomial of degree at most $2^{c'+1}\xi$; the dimension of $\bfB^{c, c'}$ is at most $(n^q/2^{c'}) \times (n^r/2^c)$, and each entry in $\bfB^{c, c'}$ is a polynomial of degree at most $2^{c+1}\xi$. Let $\Delta = 2^{c'+1}\xi$, we ``decompose'' $\bfB^{c, c'}$ into $\ell = 2^{c-c'}$ matrices $\{\bfB^{c, c', i}\}_{i=0}^{\ell-1}$, such that:
	\[\bfB^{c, c'} = \bfB^{c, c', 0} + \bfB^{c, c', 1}\cdot x^\Delta + \bfB^{c, c', 2}\cdot x^{2\Delta} + \dots + \bfB^{c, c', \ell-1}\cdot x^{(\ell-1)\Delta},\]
	and each entry in each matrix $\bfB^{c, c', i}$ has degree at most $\Delta$.
	
	We concatenate these degree-$\Delta$ matrices together, to form a matrix
	\[\widehat{\bfB^{c, c'}} = \begin{bmatrix}\bfB^{c, c', 0} & \bfB^{c, c', 1} & \ldots & \bfB^{c, c', \ell-1}\end{bmatrix}.\]
	This matrix has at most $(n^r/2^c) \cdot \ell \le (n^r/2^{c'})$ columns.
	
	Then we compute $\widehat{\bfC^{c, c'}} = \bfA^{c'}\cdot \widehat{\bfB^{c, c'}}$. We can see that
	\[\widehat{\bfC^{c, c'}} = \begin{bmatrix}\bfA^{c'}\bfB^{c, c', 0} & \bfA^{c'}\bfB^{c, c', 1} & \ldots & \bfA^{c'}\bfB^{c, c', \ell-1}\end{bmatrix}.\]
	And we can directly compute $\bfA^{c'}\cdot \bfB^{c, c'}$ from $\widehat{\bfC^{c, c'}}$, as
	\[\bfA^{c'} \cdot \bfB^{c, c'} = \sum_{i=0}^{\ell-1}\bfA^{c'}\bfB^{c, c', i}\cdot x^{i\cdot \Delta}.\]
	Now we finished the description of the algorithm.
	
	We analyze the time complexity. Fix constants $0\le c' \le c$, we need to multiply $\bfA^{c'}$ and $\widehat{\bfB^{c, c'}}$. Let $\tau = \log_n (2^{c'})$. In both of these matrices, the degree of every entry is at most $\Delta = O(2^{c'}\xi) = O(n^\tau\xi)$. The dimensions of $\bfA^{c'}$ and $\widehat{\bfB^{c, c'}}$ are upper bounded by $n^p\times (n^{q-\tau})$ and $(n^{q-\tau})\times (n^{r-\tau})$ respectively. Therefore the time complexity for this step is
	\[\tilde{O}\mleft(n^\tau\xi \cdot n^{\omega(p, q-\tau, r-\tau)}\mright),\]
	which is at most $\xi\cdot n^{\omega(p, q, r) + o(1)}$ by \cref{lemma:matmul1}. As we only need to consider $O(\log^2 n)$ pairs of $(c, c')$, it follows that the total time complexity of our algorithm is $\xi\cdot n^{\omega(p, q, r) + o(1)}$.
\end{proof}

\section{Computing Unique Shortest Paths in Directed Graphs}\label{sec:breaking-tie}

In this section, we show how to compute \emph{unique} shortest paths in a directed graph in $\tilde{O}(n^{2+\mu}M)$ time, matching the current best time bound for computing the all-pairs distances \cite{Zwick02}. Here $\mu < 0.5286$ is the solution of $\omega(1, 1, \mu) = 1 + 2\mu$~\cite{GallU18}. This algorithm is needed before we use \cref{lem:fast}.

We may assume that before we proceed, we have already computed the all-pairs distances $\|uv\|$ for every $u, v\in V$, using the APSP algorithm in \cite{Zwick02}.

Our tie-breaking method requires a (random) permutation $\pi$ of all vertices, or equivalently a bijection between the vertex set $V$ and $[n]$, i.e.~$\pi:V\to[n]$. According to $\pi$, for every graph $G$ on $V$ and every $u, v\in V$, we will specify a shortest path $\rho_G(u, v)$ in $G$ from $u$ to $v$ in a certain way. These shortest paths will be \emph{consistent} and \emph{easy to compute}, which is captured by the following theorem. (See also \cite[Theorem 1.3 and 1.4]{Ren20}.)

\begin{restatable}{theorem}{ThmUnique}\label{thm:breaking-tie}
    Given a graph $G$ on $V$, a representation of the set of shortest paths $\{\rho_G(u, v)\}_{u, v\in V}$ can be computed in $\tilde{O}(n^{2+\mu}M)$ time, with high probability over the random choice of permutation $\pi$, such that the following hold.
    \begin{enumerate}[({Property} a)]
		\item Let $G$ be a graph on $V$. For every $u', v'\in\rho_G(u, v)$ such that $u'$ appears before $v'$, the portion of $u'\To v'$ in $\rho_G(u, v)$ coincides with the path $\rho_G(u', v')$.\label{item:subpath-consistency}
		\item Let $G$ be a graph on $V$, $u, v\in V$, and $G'$ be a subgraph of $G$. Suppose $\rho_G(u, v)$ is completely contained in $G'$, then $\rho_{G'}(u, v) = \rho_G(u, v)$.\label{item:subgraph-consistency}
	\end{enumerate}
\end{restatable}

From (Property \ref{item:subpath-consistency}), for every vertex $u$, the shortest paths from $u$ to every other vertex in $G$ form a tree, and we call this tree the \emph{outgoing shortest path tree} rooted at $u$, denoted as $\Tout(u)$. Similarly, the shortest paths to $u$ from every other vertex in $G$ also form a tree, and we call this tree the \emph{incoming shortest path tree} rooted at $u$, denoted as $\Tin(u)$. Actually, the ``representation'' computed is exactly the set of $n$ outgoing shortest path trees $\{\Tout(u)\}_{u\in V}$ and the set of $n$ incoming shortest path trees $\{\Tin(u)\}_{u\in V}$.

\paragraph{The rest of this section.} We first define the paths $\rho_G(u, v)$ in \cref{sec:def-rhouv}. Then we explain how to compute them efficiently in \cref{sec:computing-rhouv}, by presenting an algorithm that computes the incoming and outgoing shortest path trees in $\tilde{O}(Mn^{2+\mu})$ time. Finally, we prove (Property~\ref{item:subpath-consistency}) and (Property~\ref{item:subgraph-consistency}) in \cref{sec:proof-of-consistency}.

\subsection{Defining $\rho_G(u, v)$}
\label{sec:def-rhouv}
Let $G$ be an input graph, and $\pi:V\to[n]$ be a (random) bijection. Let $u, v\in V$, $P$ be a path from $u$ to $v$, we will say that any vertex on $P$ that is neither $u$ nor $v$ is an \emph{internal vertex of $P$}.

Recall that we defined $|uv|$ as the \emph{largest} number of edges in any shortest path from $u$ to $v$. In particular: \begin{itemize}
	\item $|uv| = 0$ if and only if $u = v$;
	\item $|uv| = 1$ if and only if the edge $(u,v)$ is the \emph{only} shortest path from $u$ to $v$;
	\item $|uv| = \infty$ if and only if there is no path from $u$ to $v$ in $G$;
	\item otherwise, we have $2 \le |uv| < \infty$.
\end{itemize}

We claim that the set of vertices mapped to small values by $\pi$ is a good ``hitting set'' w.h.p:

\begin{claim}\label{claim:breaking-tie-hitting-set}
	Fix the graph $G$. For some large constant $C$, with high probability over the choice of $\pi$, the following holds. For every pair of vertices $u, v\in V$ such that $2 \le |uv| < \infty$, there is a shortest path $\rho'(u, v)$ from $u$ to $v$, and an internal vertex $z$ on $\rho'(u, v)$, such that $\pi(z) \le CMn\ln n/\|uv\|$.
\end{claim}
\begin{proof}
	Fix two vertices $u, v\in V$, and any shortest path $\rho'(u, v)$ from $u$ to $v$. Denote $r=\|uv\|$, if $r\le M\ln n$ then the claim is trivial. Otherwise, there are at least $r/1.1M$ vertices on $\rho'(u, v)$. Therefore, the probability over a random bijection $\pi:V\to [n]$ that $\pi$ maps every vertex on $\rho'(u, v)$ to an integer greater than $CMn\ln n/r$ is at most
	\[(1-CM\ln n/r)^{r/1.1M} \le 1/n^{C/1.1}.\]
	Thus by a union bound, the probability that the above condition holds (for every $u,v$) is at least $1-1/n^{C/1.1-2}$, which is a high probability.
\end{proof}

Let $u, v\in V$ such that $2 \le |uv| < \infty$.  Define $w(u, v)$ as the intermediate vertex with the smallest label in any shortest path from $u$ to $v$, i.e.
\begin{equation}\label{eq:def-of-wuv}
	w(u, v) = \arg_w\min\{\pi(w) : \|uv\| = \|uw\| + \|wv\|, w\ne u\text{ and }w\ne v\}.
\end{equation}

\cref{claim:breaking-tie-hitting-set} states that w.h.p.~for every vertices $u, v\in V$ such that $2 \le |uv| < \infty$, we have that
\begin{equation}\label{eq:wuv-is-small}
	\pi(w(u, v)) \le CMn\ln n / \|uv\|.
\end{equation}

In the rest of this section, we assume that \cref{eq:wuv-is-small} holds for every vertices $u, v\in V$ such that $2 \le |uv| < \infty$. Now we define the paths $\rho_G(u, v)$.

\begin{definition}
	Let $u, v\in V$ such that $|uv| \ne \infty$. The path $\rho_G(u, v)$ is recursively defined as follows. \begin{itemize}
		\item If $u = v$, then $\rho_G(u, v)$ is the empty path that starts and ends at $u$.
		\item If $|uv| = 1$, then $\rho_G(u, v)$ consists of a single edge, i.e.~the edge from $u$ to $v$.
		\item Otherwise, let $w = w(u, v)$, then $\rho_G(u, v)$ is the concatenation of $\rho_G(u, w)$ and $\rho_G(w, v)$.
	\end{itemize}
\end{definition}

For every $u, v$ such that $2 \le |uv| < \infty$, since $w$ is an intermediate vertex on some shortest path from $u$ to $v$, it is easy to see that $|uw| < |uv|$ and $|wv| < |uv|$. Therefore $\rho_G(u, v)$ is well defined --- it is inductively defined in the nondecreasing order of $|uv|$.

\subsection{Computing Shortest Path Trees in $\tilde{O}(Mn^{2+\mu})$ Time}
\label{sec:computing-rhouv}

We will need the following classical algorithm for computing distance products:
\begin{lemma}[\cite{Zwick02}]\label{lemma:distance-product-algo}
	Let $A$ be an $n\times m$ matrix, and $B$ be an $m\times n$ matrix. Suppose every entry in $A$ or $B$ is either $+\infty$ or an integer with absolute value at most $M$. Then the distance product of $A$ and $B$ can be computed in $\tilde{O}(M\cdot \MM(n, m, n))$ time.
\end{lemma}

\paragraph{Computing $w(u, v)$.} We first show how to compute $w(u, v)$ for every $u, v\in V$ such that $2 \le |uv| < \infty$ in $\tilde{O}(Mn^{2+\mu})$ time. Then we use the values of all $w(u, v)$ to compute the incoming and outgoing shortest path trees in $\tilde{O}(n^2)$ additional time. Our strategy for computing $w(u, v)$ is to mimic the algorithm in \cite{KowalukL05, ShapiraYZ11} for computing maximum witness of Boolean matrix multiplication. In particular, we divide the possible witnesses into blocks, and use fast matrix multiplication algorithms to find the block containing $w(u, v)$, for every $u, v$. After that, we use brute force to find $w(u, v)$ inside that block. Details follow.

Let $r = 2^k$ be a parameter, we show how to compute $w(u, v)$ for every pair of vertices $u, v\in V$ such that $r \le \|uv\| < 2r$. Let
\[\caH_r = \{z \in V : \pi(z) \le CMn\ln n/r\}.\]
By \cref{claim:breaking-tie-hitting-set}, for every vertices $u,v$ such that $\|uv\| \in [r, 2r)$, we have $w(u, v)\in\caH_r$. 

We define an $n\times |\caH_r|$ matrix $A$ and an $|\caH_r|\times n$ matrix $B$ as follows. For every $u\in V$ and $z\in \caH_r$, we define
\[A[u, z] = \begin{cases} \|uz\| & \text{if }\|uz\| \le 2r\text{ and }u\ne z\\ +\infty & \text{otherwise}\end{cases},\text{ and }
B[z, u] = \begin{cases}\|zu\| & \text{if }\|zu\| \le 2r\text{ and }u\ne z\\ +\infty & \text{otherwise}\end{cases}.\]
Then we compute the \emph{minimum witness} of the distance product $A\star B$. To be more precise, we compute the matrix $W[\cdot, \cdot]$ such that for every $u, v\in V$,
\[W[u, v] = \arg_z\min \{\pi(z) : \|uv\| = A[u, z] + B[z, v]\}.\]

\subparagraph{Correctness.} Fix $u, v\in V$, where $\|uv\| \in [r, 2r)$. We will show that if $|uv| = 1$, then $W[u, v]$ does not exist; otherwise $W[u, v]$ coincides with $w(u, v)$ defined in \cref{eq:def-of-wuv}.

First, suppose $|uv| = 1$, then there are no intermediate vertex $z$ such that $\|uv\| = \|uz\| + \|zv\|$, which means $W[u, v]$ does not exist.

Now we assume $|uv| \ge 2$. Since $\|uv\| \ge r$, by \cref{claim:breaking-tie-hitting-set}, there is an intermediate vertex $z\in\caH_r$ such that $\|uz\| + \|zv\| = \|uv\|$. Since $\|uz\|, \|zv\| \le \|uv\| < 2r$, we can see that $\|uv\| = A[u, z] + B[z, v]$, therefore $W[u, v]$ exists. Let $z = W[u, v]$, then by \cref{eq:def-of-wuv}, $\pi(w(u, v)) \le \pi(z)$. On the other hand, \cref{claim:breaking-tie-hitting-set} shows that $w(u, v)\in\caH_r$, so by the definition of $z = W[u, v]$, we have $\pi(z) \le \pi(w(u, v))$. Therefore $z = w(u, v)$ and we have established the correctness of $W[\cdot, \cdot]$.

\subparagraph{Time complexity.} Now we show how to compute the matrix $W[\cdot, \cdot]$ efficiently.

Let $s = n^\mu$, where $\mu \in (0, 1)$ is a parameter to be determined later. If $|\caH_r| < s$, then we can compute the matrix $W$ by brute force in $\tilde{O}(n^2s)$ time. Otherwise, we partition $\caH_r$ into blocks of size $s$, where the $i$-th block contains vertices that are mapped by $\pi$ to values between $(i-1)\cdot s+1$ and $i\cdot s$. For every block $i$, we compute the distance product of $A$ and $B$ where only vertices in block $i$ are allowed as witnesses. In other words, we compute the following matrix
\[D^i[u, v] = \min\{A[u, z] + B[z, v] : (i-1)\cdot s+1 \le \pi(z) \le i\cdot s\}.\]
By \cref{lemma:distance-product-algo}, this matrix can be computed in $\tilde{O}(r\cdot \MM(n, s, n))$ time. There are $O(|\caH_r|/s)=\tilde{O}(Mn/(rs))$ blocks, and we need to compute a distance product $D^i$ for each block $i$. Therefore the total time for computing all these distance products is
\[\tilde{O}(r\cdot \MM(n, s, n) \cdot Mn/(rs)) = \tilde{O}(M\cdot (n/s)\cdot \MM(n, s, n)).\]

Now for every $u, v\in V$ such that $\|uv\| \in [r, 2r)$ and $|uv|\ge 2$, we want to compute $W[u, v]$, which is the vertex $z\in\caH_r$ with the minimum $\pi(z)$, such that $\|uv\| = A[u, z] + B[z, v]$. First, we find the smallest $i$ such that $D^i[u, v] = \|uv\|$, and we know that $W[u, v]$ is in the $i$-th block. (If such $i$ does not exist, then $W[u, v]$ does not exist either, and $|uv| = 1$.) This step takes $\tilde{O}(Mn/(rs))$ time. Then we iterate through the vertices in this block, and find the vertex $z$ with the smallest $\pi(z)$ such that $A[u, z] + B[z, v] = \|uv\|$. This step takes $O(s)$ time.

It follows that the time complexity for computing every $w(u, v)$ where $\|uv\| \in [r, 2r)$ is
\begin{align}
&\,\tilde{O}(M\cdot \MM(n, s, n) \cdot (n/s) + n^2\cdot Mn/(rs) + n^2s)\nonumber\\
\le&\,\tilde{O}(M\cdot\MM(n, s, n)\cdot (n/s) + n^2s)\label{eq:step1}\\
\le&\,\tilde{O}(M\cdot n^{\omega(1, \mu, 1) + 1 - \mu} + n^{2+\mu})\nonumber.
\end{align}

Here, \cref{eq:step1} is because $n^2\cdot Mn/(rs) \le n^2\cdot M \cdot (n/s)\le M\cdot \MM(n, s, n)\cdot (n/s)$.

Let $\mu$ be the solution to $\omega(1, \mu, 1) = 1 + 2\mu$, then $\mu < 0.5286$ (\cite{Zwick02, GallU18}). It follows that the time complexity for computing every $w(u, v)$, where $r \le \|uv\| < 2r$, is at most $\tilde{O}(Mn^{2 + \mu})$.

\subparagraph{Putting it together.} We run the above algorithm for $k$ from $0$ to $\lfloor\log (nW)\rfloor$, and for each $k$, we update the values $w(u, v)$ where $\|uv\| \in [2^k, 2^{k+1})$. The total time to compute $w(u, v)$ for all $u, v$ is thus $\tilde{O}(Mn^{2 + \mu})$.

\def\parent{\mathsf{parent}}
\paragraph{From $w(u, v)$ to unique shortest paths.} For every $u, v\in V$, we will compute the parent of $u$ in the tree $\Tin(v)$, denoted as $\parent_v(u)$. In other words, $\parent_v(u)$ is the second vertex in the path $\rho_G(u, v)$ (the first being $u$). After computing $\parent_v(u)$ for every $u, v\in V$, it is easy to construct $\Tin(v)$ for every vertex $v$. We can compute every $\Tout(u)$ in a symmetric fashion.

We proceed by nondecreasing order of $\|uv\|$. Suppose that for every $(u', v')$ such that $\|u'v'\| < \|uv\|$, we have already computed $\parent_{v'}(u')$. Now we compute $\parent_v(u)$ as follows. Let $w = w(u, v)$. If $w$ does not exist, let $\parent_v(u) = v$; otherwise $\parent_v(u) = \parent_w(u)$. 

This algorithm (that given every $w(u, v)$, computes every $\parent_v(u)$) clearly runs in $\tilde{O}(n^2)$ time. Notice that if $w$ exists, then $w$ is an intermediate vertex in $\rho_G(u, v)$, thus $\|uw\| < \|uv\|$, and the second vertex in the path $\rho_G(u, v)$ coincides with the second vertex in the path $\rho_G(u, w)$. Hence, the correctness of the algorithm can be easily proved by induction on $\|uv\|$.

\subsection{Proof of \cref{thm:breaking-tie}}
\label{sec:proof-of-consistency}
\ThmUnique*
In this subsection, for any path $P$ and vertices $u',v'\in P$ such that $u'$ appears before $v'$ on $P$, we use $P[u',v']$ to denote the portion of $u'\To v'$ on the path $P$.

\begin{proof}[Proof of (Property \ref{item:subpath-consistency})]
    We prove it by induction on the number of edges of $\rho_G(u,v)$. Let $P=\rho_G(u, v)$. If $u=v$ or $P$ has only one edge, (Property \ref{item:subpath-consistency}) is trivial. Now suppose $P$ has $k$ edges where $k > 1$. Let $w = w(u, v)$, then $w$ must lie on $P$. Consider the following three cases:\begin{itemize}
        \item Suppose $u'$ appears after (or coincides with) $w$ on $P$. By definition, $P[w,v] = \rho_G(w,v)$. Then $P[u',v'] = \rho_G(u',v')$ by induction hypothesis on $\rho_G(w,v)$ since it has fewer edges than $\rho_G(u,v)$.
        \item Suppose $v'$ appears before (or coincides with) $w$. This case is symmetric to the above case.
        \item Otherwise, $w$ lies between $u'$ and $v'$ on $P$.
        
        First, we claim that $w = w(u', v')$. As $w$ lies on some shortest path from $u'$ to $v'$ (i.e.~$P[u', v']$), we have $\pi(w(u', v')) \le \pi(w)$. On the other hand, suppose there exists $w'$ such that $\pi(w')<\pi(w)$ and $w'$ is on some shortest path from $u'$ to $v'$. Then $w'$ also lies on some shortest path from $u$ to $v$, so it is a better candidate for $w(u, v)$, contradicting the definition of $w$.
        
        Second, by induction hypothesis on $\rho_G(u, w)$, which has fewer edges than $\rho_G(u, v)$, we have $P[u',w] = \rho_G(u',w)$. Similarly, $P[w,v']=\rho_G(w, v')$. Therefore, by definition, $P[u',v']=P[u',w]\circ P[w,v']=\rho_G(u',v')$.\qedhere
    \end{itemize}
\end{proof}

\begin{proof}[Proof of (Property \ref{item:subgraph-consistency})]
    We prove it by induction on the number of edges of $\rho_G(u,v)$. Let $P=\rho_G(u, v)$. If $u=v$ or $P$ has only one edge, (Property \ref{item:subgraph-consistency}) is trivial.
    
    Now suppose $P$ has more than one edge. Let $w=w_G(u,v)$ (i.e.~the vertex $w(u, v)$ defined in \cref{eq:def-of-wuv} in graph $G$), we claim that $w$ coincides with $w_{G'}(u,v)$ (i.e.~the vertex $w(u, v)$ defined in \cref{eq:def-of-wuv} in graph $G'$). Since $P$ is also a shortest path from $u$ to $v$ in $G'$, we have $\pi(w_{G'}(u, v)) \le \pi(w)$. On the other hand, suppose there exists $w'$ such that $\pi(w') < \pi(w)$ and $w'$ is on some shortest path from $u$ to $v$ in $G'$. Then $w'$ also lies on some shortest path from $u$ to $v$ in $G$, so it is a better candidate for $w_G(u,v)$, contradicting the definition of $w$.
    
    Since $\rho_G(u, w)$ has fewer edges than $\rho_G(u, v)$, and $\rho_G(u, w)$ is completely contained in $G'$, we can use induction hypothesis on $\rho_G(u, w)$ to conclude that $P[u, w]=\rho_{G'}(u, w)$. Similarly, we can use the induction hypothesis on $\rho_G(w, v)$ to conclude that $P[w, v]=\rho_{G'}(w, v)$. Therefore, by definition, $\rho_{G'}(u,v) = \rho_{G'}(u, w) \circ \rho_{G'}(w, v) = P$.
\end{proof}
\section{Conclusions and Open Problems}
We presented an improved DSO for directed graphs with integer weights in $[1, M]$. The preprocessing time is $O(n^{2.5794}M)$ and the query time is $O(1)$. However, there is still a small gap between the preprocessing time of our DSO and the current best time bound for the APSP problem in directed graphs, which is $\tilde{O}(n^{2+\mu}M) \le O(n^{2.5286}M)$ \cite{Zwick02}. Can we improve the preprocessing time to $\tilde{O}(n^{2+\mu} M)$, matching the latter time bound? Another interesting problem is to investigate the complexity of preprocessing a DSO in undirected graphs --- here, the best time bound for APSP is $\tilde{O}(n^\omega M)$ \cite{Seidel95, ShoshanZ99}. Can we preprocess a DSO in $\tilde{O}(n^\omega M)$ time on undirected graphs?

Compared to other DSOs \cite{WeimannY13, GrandoniW12, ChechikC20}, our oracle has two drawbacks. First, our query algorithm only outputs the shortest distance, but we do not know how to find the actual shortest paths. So another open problem is whether we can find the actual shortest path with additional $O(l)$ query time, where $l$ is the number of edges in the returned shortest path. Second, since we used \cite[Observation 2.1]{Ren20}, our oracle can only deal with positive edge weights. Can we extend our oracle to also deal with negative edge weights?

For every parameter $f$, the $r$-truncated DSO in \cref{sec:r-truncated-DSO} can actually handle $f$ edge/vertex deletions in $\Tilde{O}(f^\omega r)$ query time. (See also \cite{vdBS19}.) However, as far as we know, \cite[Observation 2.1]{Ren20} only works for one failure. It would be exciting to extend \cite[Observation 2.1]{Ren20} or our (full) DSO to also handle $f$ failures.

\section*{Acknowledgment}
We thank Ran Duan and Tianyi Zhang for their helpful discussions during the initial stage of this research. We are grateful to anonymous reviewers for their helpful comments. We would also like to thank an anonymous reviewer for suggesting the title of \cref{sec:breaking-tie}, and another anonymous reviewer for pointing out a subtle issue regarding the invertibility of polynomial matrices (and fixing the issue).

\bibliography{main}	

\newcommand{\proc}{Proc. } \newcommand{\stoc}[1]{\proc #1 Annual ACM Symposium
  on Theory of Computing ({STOC})} \newcommand{\focs}[1]{\proc #1 Annual IEEE
  Symposium on Foundations of Computer Science ({FOCS})}
  \newcommand{\ccc}[1]{\proc #1 Annual IEEE Conference on Computational
  Complexity ({CCC})} \newcommand{\sct}[1]{\proc #1 Annual Structure in
  Complexity Theory Conference} \newcommand{\icalp}[1]{\proc #1 International
  Colloquium on Automata, Languages and Programming ({ICALP})}
  \newcommand{\soda}[1]{\proc #1 Annual ACM-SIAM Symposium on Discrete
  Algorithms ({SODA})} \newcommand{\apprx}[1]{\proc #1 International Workshop
  on Approximation Algorithms for Combinatorial Optimization Problems
  ({APPROX})} \newcommand{\rnd}[1]{\proc #1 International Workshop on
  Randomization and Approximation Techniques in Computer Science ({RANDOM})}
  \newcommand{\fsttcs}[1]{\proc #1 Annual Conference on Foundations of Software
  Technology and Theoretical Computer Science ({FSTTCS})}
  \newcommand{\mfcs}[1]{\proc #1 International Symposium on Mathematical
  Foundations of Computer Science ({MFCS})} \newcommand{\itcs}[1]{\proc #1
  {C}onference on {I}nnovations in {T}heoretical {C}omputer {S}cience ({ITCS})}
  \newcommand{\sosa}[1]{\proc #1 Symposium on Simplicity in Algorithms
  ({SOSA})} \newcommand{\issac}[1]{\proc #1 International Symposium on Symbolic
  and Algebraic Computation, ({ISSAC})} \newcommand{\esa}[1]{\proc #1 European
  Symposium on Algorithms ({ESA})} \newcommand{\tcc}[1]{\proc #1 Theory of
  Cryptography Conference ({TCC})} \newcommand{\csr}[1]{\proc #1 International
  Computer Science Symposium in Russia, ({CSR})} \newcommand{\cocoon}[1]{\proc
  #1 International Computing and Combinatorics Conference ({COCOON})}
  \newcommand{\stacs}[1]{\proc #1 Symposium on Theoretical Aspects of Computer
  Science, ({STACS})} \newcommand{\wads}[1]{\proc #1 International Symposium on
  Algorithms and Data Structures ({WADS})} \newcommand{\latin}[1]{\proc #1
  Latin American Theoretical Informatics Symposium, ({LATIN})}
  \newcommand{\eccc}{Electronic Colloquium on Computational Complexity: {ECCC}}
  \newcommand{\jacm}{Journal of the ACM} \newcommand{\coco}{Computational
  Complexity} \newcommand{\jcss}{Journal of Computer and System Sciences}
  \newcommand{\siamj}{{SIAM} Journal of Computing}
  \newcommand{\ipl}{Information Processing Letters} \newcommand{\tocj}{Theory
  of Computing} \newcommand{\tcs}{Theoretical Computer Science}
  \newcommand{\toct}{ACM Transactions on Computation Theory}
  \newcommand{\toit}{IEEE Transactions on Information Theory}
  \newcommand{\talg}{ACM Transactions on Algorithms}
\begin{thebibliography}{vdBNS19}

\bibitem[ACC19]{AlonCC19}
Noga Alon, Shiri Chechik, and Sarel Cohen.
\newblock Deterministic combinatorial replacement paths and distance
  sensitivity oracles.
\newblock In {\em \icalp{46th}}, volume 132 of {\em LIPIcs}, pages 12:1--12:14,
  2019.
\newblock \href {https://doi.org/10.4230/LIPIcs.ICALP.2019.12}
  {\path{doi:10.4230/LIPIcs.ICALP.2019.12}}.

\bibitem[AGM97]{AlonGM97}
Noga Alon, Zvi Galil, and Oded Margalit.
\newblock On the exponent of the all pairs shortest path problem.
\newblock {\em \jcss}, 54(2):255--262, 1997.
\newblock \href {https://doi.org/10.1006/jcss.1997.1388}
  {\path{doi:10.1006/jcss.1997.1388}}.

\bibitem[AHU74]{AhoHU74}
Alfred~V. Aho, John~E. Hopcroft, and Jeffrey~D. Ullman.
\newblock {\em The Design and Analysis of Computer Algorithms}.
\newblock Addison-Wesley, 1974.

\bibitem[AW21]{AlmanW21}
Josh Alman and Virginia~Vassilevska Williams.
\newblock A refined laser method and faster matrix multiplication.
\newblock In {\em \soda{32nd}}, pages 522--539, 2021.
\newblock \href {https://doi.org/10.1137/1.9781611976465.32}
  {\path{doi:10.1137/1.9781611976465.32}}.

\bibitem[BH74]{MatInv}
James~R. Bunch and John~E. Hopcroft.
\newblock Triangular factorization and inversion by fast matrix multiplication.
\newblock {\em Mathematics of Computation}, 28(125):231--236, 1974.
\newblock \href {https://doi.org/10.2307/2005828} {\path{doi:10.2307/2005828}}.

\bibitem[BK08]{BernsteinK08}
Aaron Bernstein and David~R. Karger.
\newblock Improved distance sensitivity oracles via random sampling.
\newblock In {\em \soda{19th}}, pages 34--43, 2008.
\newblock URL: \url{http://dl.acm.org/citation.cfm?id=1347082.1347087}.

\bibitem[BK09]{BernsteinK09}
Aaron Bernstein and David~R. Karger.
\newblock A nearly optimal oracle for avoiding failed vertices and edges.
\newblock In {\em \stoc{41st}}, pages 101--110, 2009.
\newblock \href {https://doi.org/10.1145/1536414.1536431}
  {\path{doi:10.1145/1536414.1536431}}.

\bibitem[Bl{\"{a}}13]{Blaser13}
Markus Bl{\"{a}}ser.
\newblock Fast matrix multiplication.
\newblock {\em Theory of Computing, Graduate Surveys}, 5:1--60, 2013.
\newblock \href {https://doi.org/10.4086/toc.gs.2013.005}
  {\path{doi:10.4086/toc.gs.2013.005}}.

\bibitem[CC20]{ChechikC20}
Shiri Chechik and Sarel Cohen.
\newblock Distance sensitivity oracles with subcubic preprocessing time and
  fast query time.
\newblock In {\em \stoc{52nd}}, pages 1375--1388, 2020.
\newblock \href {https://doi.org/10.1145/3357713.3384253}
  {\path{doi:10.1145/3357713.3384253}}.

\bibitem[CW90]{CW90}
Don Coppersmith and Shmuel Winograd.
\newblock Matrix multiplication via arithmetic progressions.
\newblock {\em Journal of Symbolic Computation}, 9(3):251--280, 1990.
\newblock \href {https://doi.org/10.1016/S0747-7171(08)80013-2}
  {\path{doi:10.1016/S0747-7171(08)80013-2}}.

\bibitem[DI04]{DemetrescuI04}
Camil Demetrescu and Giuseppe~F. Italiano.
\newblock A new approach to dynamic all pairs shortest paths.
\newblock {\em \jacm}, 51(6):968--992, 2004.
\newblock \href {https://doi.org/10.1145/1039488.1039492}
  {\path{doi:10.1145/1039488.1039492}}.

\bibitem[DP09a]{DuanP09a}
Ran Duan and Seth Pettie.
\newblock Dual-failure distance and connectivity oracles.
\newblock In {\em \soda{20th}}, pages 506--515, 2009.
\newblock \href {https://doi.org/10.1137/1.9781611973068.56}
  {\path{doi:10.1137/1.9781611973068.56}}.

\bibitem[DP09b]{DuanP09}
Ran Duan and Seth Pettie.
\newblock Fast algorithms for (max, min)-matrix multiplication and bottleneck
  shortest paths.
\newblock In {\em \soda{20th}}, pages 384--391, 2009.
\newblock \href {https://doi.org/10.1137/1.9781611973068.43}
  {\path{doi:10.1137/1.9781611973068.43}}.

\bibitem[DTCR08]{DemetrescuTCR08}
Camil Demetrescu, Mikkel Thorup, Rezaul~Alam Chowdhury, and Vijaya
  Ramachandran.
\newblock Oracles for distances avoiding a failed node or link.
\newblock {\em \siamj}, 37(5):1299--1318, 2008.
\newblock \href {https://doi.org/10.1137/S0097539705429847}
  {\path{doi:10.1137/S0097539705429847}}.

\bibitem[DZ17]{DuanZ17a}
Ran Duan and Tianyi Zhang.
\newblock Improved distance sensitivity oracles via tree partitioning.
\newblock In {\em \wads{15th}}, volume 10389 of {\em LNCS}, pages 349--360,
  2017.
\newblock \href {https://doi.org/10.1007/978-3-319-62127-2\_30}
  {\path{doi:10.1007/978-3-319-62127-2\_30}}.

\bibitem[GU18]{GallU18}
Francois~Le Gall and Florent Urrutia.
\newblock Improved rectangular matrix multiplication using powers of the
  {C}oppersmith-{W}inograd tensor.
\newblock In {\em \soda{29th}}, pages 1029--1046, 2018.
\newblock \href {https://doi.org/10.1137/1.9781611975031.67}
  {\path{doi:10.1137/1.9781611975031.67}}.

\bibitem[GW20]{GrandoniW12}
Fabrizio Grandoni and Virginia~Vassilevska Williams.
\newblock Faster replacement paths and distance sensitivity oracles.
\newblock {\em \talg}, 16(1):15:1--15:25, 2020.
\newblock \href {https://doi.org/10.1145/3365835} {\path{doi:10.1145/3365835}}.

\bibitem[KL05]{KowalukL05}
Miroslaw Kowaluk and Andrzej Lingas.
\newblock {LCA} queries in directed acyclic graphs.
\newblock In {\em \icalp{32nd}}, volume 3580 of {\em LNCS}, pages 241--248,
  2005.
\newblock \href {https://doi.org/10.1007/11523468\_20}
  {\path{doi:10.1007/11523468\_20}}.

\bibitem[LG14]{LeGall}
Fran\c{c}ois Le~Gall.
\newblock Powers of tensors and fast matrix multiplication.
\newblock In {\em \issac{39th}}, pages 296--303, 2014.
\newblock \href {https://doi.org/10.1145/2608628.2608664}
  {\path{doi:10.1145/2608628.2608664}}.

\bibitem[LNZ17]{labahn2017fast}
George Labahn, Vincent Neiger, and Wei Zhou.
\newblock Fast, deterministic computation of the {H}ermite normal form and
  determinant of a polynomial matrix.
\newblock {\em Journal of Complexity}, 42:44--71, 2017.
\newblock \href {https://doi.org/10.1016/j.jco.2017.03.003}
  {\path{doi:10.1016/j.jco.2017.03.003}}.

\bibitem[LR83]{LottiR83}
Grazia Lotti and Francesco Romani.
\newblock On the asymptotic complexity of rectangular matrix multiplication.
\newblock {\em \tcs}, 23:171--185, 1983.
\newblock \href {https://doi.org/10.1016/0304-3975(83)90054-3}
  {\path{doi:10.1016/0304-3975(83)90054-3}}.

\bibitem[LWW18]{LincolnWW18}
Andrea Lincoln, Virginia~Vassilevska Williams, and R.~Ryan Williams.
\newblock Tight hardness for shortest cycles and paths in sparse graphs.
\newblock In {\em \soda{29th}}, pages 1236--1252, 2018.
\newblock \href {https://doi.org/10.1137/1.9781611975031.80}
  {\path{doi:10.1137/1.9781611975031.80}}.

\bibitem[Ren20]{Ren20}
Hanlin Ren.
\newblock Improved distance sensitivity oracles with subcubic preprocessing
  time.
\newblock In {\em \esa{28th}}, volume 173 of {\em LIPIcs}, pages 79:1--79:13,
  2020.
\newblock \href {https://doi.org/10.4230/LIPIcs.ESA.2020.79}
  {\path{doi:10.4230/LIPIcs.ESA.2020.79}}.

\bibitem[San05a]{Sankowski05}
Piotr Sankowski.
\newblock Shortest paths in matrix multiplication time.
\newblock In {\em \esa{13th}}, volume 3669 of {\em LNCS}, pages 770--778, 2005.
\newblock \href {https://doi.org/10.1007/11561071\_68}
  {\path{doi:10.1007/11561071\_68}}.

\bibitem[San05b]{Sankowski05-dynamic}
Piotr Sankowski.
\newblock Subquadratic algorithm for dynamic shortest distances.
\newblock In {\em \cocoon{11th}}, volume 3595 of {\em LNCS}, pages 461--470,
  2005.
\newblock \href {https://doi.org/10.1007/11533719\_47}
  {\path{doi:10.1007/11533719\_47}}.

\bibitem[Sch80]{Schwartz80}
Jacob~T. Schwartz.
\newblock Fast probabilistic algorithms for verification of polynomial
  identities.
\newblock {\em \jacm}, 27(4):701--717, 1980.
\newblock \href {https://doi.org/10.1145/322217.322225}
  {\path{doi:10.1145/322217.322225}}.

\bibitem[Sei95]{Seidel95}
Raimund Seidel.
\newblock On the all-pairs-shortest-path problem in unweighted undirected
  graphs.
\newblock {\em \jcss}, 51(3):400--403, 1995.
\newblock \href {https://doi.org/10.1006/jcss.1995.1078}
  {\path{doi:10.1006/jcss.1995.1078}}.

\bibitem[SM50]{ShermanM50}
Jack Sherman and Winifred~J. Morrison.
\newblock Adjustment of an inverse matrix corresponding to a change in one
  element of a given matrix.
\newblock {\em The Annals of Mathematical Statistics}, 21(1):124--127, 1950.
\newblock URL: \url{http://www.jstor.org/stable/2236561}.

\bibitem[SP21]{KarthikP21}
Karthik~C. S. and Merav Parter.
\newblock Deterministic replacement path covering.
\newblock In {\em \soda{32nd}}, pages 704--723, 2021.
\newblock \href {https://doi.org/10.1137/1.9781611976465.44}
  {\path{doi:10.1137/1.9781611976465.44}}.

\bibitem[Sto03]{Storjohann03}
Arne Storjohann.
\newblock High-order lifting and integrality certification.
\newblock {\em Journal of Symbolic Computation}, 36(3-4):613--648, 2003.
\newblock \href {https://doi.org/10.1016/S0747-7171(03)00097-X}
  {\path{doi:10.1016/S0747-7171(03)00097-X}}.

\bibitem[Sto10]{Sto10}
Andrew~James Stothers.
\newblock {\em On the complexity of matrix multiplication}.
\newblock PhD thesis, The University of Edinburgh, 2010.

\bibitem[SYZ11]{ShapiraYZ11}
Asaf Shapira, Raphael Yuster, and Uri Zwick.
\newblock All-pairs bottleneck paths in vertex weighted graphs.
\newblock {\em Algorithmica}, 59(4):621--633, 2011.
\newblock \href {https://doi.org/10.1007/s00453-009-9328-x}
  {\path{doi:10.1007/s00453-009-9328-x}}.

\bibitem[SZ99]{ShoshanZ99}
Avi Shoshan and Uri Zwick.
\newblock All pairs shortest paths in undirected graphs with integer weights.
\newblock In {\em \focs{40th}}, pages 605--615, 1999.
\newblock \href {https://doi.org/10.1109/SFFCS.1999.814635}
  {\path{doi:10.1109/SFFCS.1999.814635}}.

\bibitem[vdBN19]{BrandN19}
Jan van~den Brand and Danupon Nanongkai.
\newblock Dynamic approximate shortest paths and beyond: Subquadratic and
  worst-case update time.
\newblock In {\em \focs{60th}}, pages 436--455, 2019.
\newblock \href {https://doi.org/10.1109/FOCS.2019.00035}
  {\path{doi:10.1109/FOCS.2019.00035}}.

\bibitem[vdBNS19]{BrandNS19}
Jan van~den Brand, Danupon Nanongkai, and Thatchaphol Saranurak.
\newblock Dynamic matrix inverse: Improved algorithms and matching conditional
  lower bounds.
\newblock In {\em \focs{60th}}, pages 456--480, 2019.
\newblock \href {https://doi.org/10.1109/FOCS.2019.00036}
  {\path{doi:10.1109/FOCS.2019.00036}}.

\bibitem[vdBS19]{vdBS19}
Jan van~den Brand and Thatchaphol Saranurak.
\newblock Sensitive distance and reachability oracles for large batch updates.
\newblock In {\em \focs{60th}}, pages 424--435, 2019.
\newblock \href {https://doi.org/10.1109/FOCS.2019.00034}
  {\path{doi:10.1109/FOCS.2019.00034}}.

\bibitem[Wil12]{Wil12}
Virginia~Vassilevska Williams.
\newblock Multiplying matrices faster than {C}oppersmith-{W}inograd.
\newblock In {\em \stoc{44th}}, pages 887--898, 2012.
\newblock \href {https://doi.org/10.1145/2213977.2214056}
  {\path{doi:10.1145/2213977.2214056}}.

\bibitem[Woo50]{Woodbury50}
Max~A Woodbury.
\newblock Inverting modified matrices.
\newblock {\em Memorandum report}, 42(106):336, 1950.

\bibitem[WY13]{WeimannY13}
Oren Weimann and Raphael Yuster.
\newblock Replacement paths and distance sensitivity oracles via fast matrix
  multiplication.
\newblock {\em \talg}, 9(2):14:1--14:13, 2013.
\newblock \href {https://doi.org/10.1145/2438645.2438646}
  {\path{doi:10.1145/2438645.2438646}}.

\bibitem[Zip79]{Zippel79}
Richard Zippel.
\newblock Probabilistic algorithms for sparse polynomials.
\newblock In {\em Symbolic and Algebraic Computation, {EUROSAM} '79}, volume~72
  of {\em LNCS}, pages 216--226, 1979.
\newblock \href {https://doi.org/10.1007/3-540-09519-5\_73}
  {\path{doi:10.1007/3-540-09519-5\_73}}.

\bibitem[ZLS12]{ZhouLS12}
Wei Zhou, George Labahn, and Arne Storjohann.
\newblock Computing minimal nullspace bases.
\newblock In {\em \issac{37th}}, pages 366--373, 2012.
\newblock \href {https://doi.org/10.1145/2442829.2442881}
  {\path{doi:10.1145/2442829.2442881}}.

\bibitem[ZLS15]{ZhouLS15}
Wei Zhou, George Labahn, and Arne Storjohann.
\newblock A deterministic algorithm for inverting a polynomial matrix.
\newblock {\em Journal of Complexity}, 31(2):162--173, 2015.
\newblock \href {https://doi.org/10.1016/j.jco.2014.09.004}
  {\path{doi:10.1016/j.jco.2014.09.004}}.

\bibitem[Zwi02]{Zwick02}
Uri Zwick.
\newblock All pairs shortest paths using bridging sets and rectangular matrix
  multiplication.
\newblock {\em \jacm}, 49(3):289--317, 2002.
\newblock \href {https://doi.org/10.1145/567112.567114}
  {\path{doi:10.1145/567112.567114}}.

\end{thebibliography}
\appendix

\section{Omitted Proofs in \cref{sec:preliminaries}}\label{sec:apd-FMM}

\MatMulI*
\begin{proof}
	The proof is adapted from \cite[Lemma 7.7]{Blaser13}; readers familiar with tensor and tensor rank may refer to the proof of that lemma.
	
	Let $g(n) = \MM(n^a, n^{b+r}, n^{c+r})$, then it is easy to see that in $g(n)$ operations we can compute $n^r$ matrix multiplication instances of size $n^a\times n^b\times n^c$. We will use induction to prove that for every integer $k$, $n^r$ matrix multiplication instances of size $n^{ka}\times n^{kb}\times n^{kc}$ can be computed in $\lceil g(n)/n^r\rceil^k\cdot n^r$ operations.
	
	The case for $k=1$ is trivial. When $k > 1$, we can compute $n^r$ matrix multiplication instances of size $n^{ka}\times n^{kb}\times n^{kc}$ as follows. First, we partition every size-$(n^{ka}\times n^{kb})$ matrix into size-$(n^{(k-1)a}\times n^{(k-1)b})$ blocks, and partition every size-$(n^{kb}\times n^{kc})$ matrix into size-$(n^{(k-1)b}\times n^{(k-1)c})$ blocks. Then we can reduce the problem to computing $n^r$ matrix multiplication instances of size $n^a\times n^b\times n^c$ using ``big operations'', where each ``big operation'' is a matrix multiplication instance of size $n^{(k-1)a}\times n^{(k-1)b}\times n^{(k-1)c}$. It suffices to perform $g(n)$ ``big operations''. On the other hand, by the induction hypothesis, we can perform each $n^r$ ``big operations'' in $\lceil g(n)/n^r\rceil^{k-1}\cdot n^r$ operations. By partitioning these $g(n)$ ``big operations'' into groups of size $n^r$, we can compute all these ``big operations'' in $\lceil g(n)/n^r\rceil\cdot \lceil g(n)/n^r\rceil^{k-1}\cdot n^r$ operations, and we are done.
	
	Now it is easy to see that
	\[\omega(a, b, c) \le \inf_{n, k}\mleft\{\log_{n^k}\mleft(\lceil g(n)/n^r\rceil^k\cdot n^r\mright)\mright\}\le \inf_n\mleft\{\log_n\lceil g(n)/n^r\rceil\mright\}=\omega(a, b+r, c+r) - r.\qedhere\]
\end{proof}

\MatMulII*
\begin{proof}
	Let $0\le \tau_1 < \tau_2 \le 1$. Then:\begin{itemize}
		\item By \cref{lemma:matmul1}, $(\tau_2-\tau_1) + \omega(1, 1-\tau_2, 1-\tau_2) \le \omega(1, 1-\tau_1, 1-\tau_1)$, which means $\tau_1 + f(\tau_1) \ge \tau_2 + f(\tau_2)$.
		\item For every integer $n$, we can compute the product of an $n\times n^{1-\tau_1}$ matrix and an $n^{1-\tau_1}\times n^{1-\tau_1}$ matrix, by using $n^{2(\tau_2-\tau_1)}$ invocations of multiplication algorithms for matrices of dimension $n\times n^{1-\tau_2}$ and $n^{1-\tau_2}\times n^{1-\tau_2}$. Therefore $\omega(1, 1-\tau_1, 1-\tau_1) \le 2(\tau_2 - \tau_1) + \omega(1, 1-\tau_2, 1-\tau_2)$, which means $2\tau_1 + f(\tau_1) \le 2\tau_2 + f(\tau_2)$.\qedhere
	\end{itemize}
\end{proof}

\end{document}